\documentclass[11pt,pdfa,letterpaper]{article}

\usepackage[in]{fullpage}

\usepackage{comment}
\usepackage{amsmath}
\usepackage{mathtools,amsthm,amssymb} %
\usepackage{xcolor,xstring,xparse}
\usepackage{graphicx}
\usepackage{comment}
\usepackage{qtree}
\usepackage{tree-dvips}
\usepackage{float}
\definecolor{linkblue}{HTML}{001487}
\usepackage[colorlinks=true,allcolors=linkblue]{hyperref}
\usepackage[nameinlink,capitalize,noabbrev]{cleveref}
\usepackage{braket}
\usepackage{mathrsfs}
\usepackage{tikz}
\usepackage{qcircuit}
\usepackage{xspace}
\usepackage{dsfont}
\usepackage{enumitem} %
\usepackage{longfbox}
\usepackage{csquotes}
\usepackage{authblk}
\usepackage{longfbox}
\usepackage{ytableau}

\usepackage{adjustbox}
\usepackage{siunitx}
\usepackage{stmaryrd}
\usepackage{float}
\usepackage{amsthm}
\usepackage{tikz-cd}
\usepackage{tikz}
\usetikzlibrary{backgrounds}
\usepackage{hyperref}

\hypersetup{
    colorlinks = true,
    linkcolor = violet,
    linkbordercolor = {white}
}

\usepackage{booktabs}    
\usepackage{tabularx}    

\usepackage{geometry}
\usepackage[style=alphabetic,minalphanames=3,maxalphanames=4,maxnames=99,backref=true]{biblatex}

\usepackage{url}
\usepackage{bbm}
\usepackage{hyperref}
\hypersetup{breaklinks=true}

\usetikzlibrary{fadings}
\usetikzlibrary{patterns}
\usetikzlibrary{shadows.blur}
\usetikzlibrary{shapes}

\usepackage{framed}

\usepackage{caption}
\newif\ifcomments
\commentstrue        

\ifcomments
  \newcommand{\andrew}[1]{{\color{purple}\bfseries (Andrew: #1)}}
  \newcommand{\akshar}[1]{{\color{red}\bfseries (Akshar: #1)}}
  \newcommand{\john}[1]{{\color{blue}\bfseries (John: #1)}}
\else
  \newcommand{\andrew}[1]{}
  \newcommand{\akshar}[1]{}
  \newcommand{\john}[1]{}
\fi

\theoremstyle{plain}
\newtheorem{theorem}{Theorem}[section]
\newtheorem*{theorem*}{Theorem}
\newtheorem{lemma}[theorem]{Lemma}
\newtheorem{claim}[theorem]{Claim}
\Crefname{claim}{Claim}{Claims}

\newtheorem*{lemma*}{Lemma}
\newtheorem{fact}[theorem]{Fact}
\newtheorem{corollary}[theorem]{Corollary}
\newtheorem*{corollary*}{Corollary}

\theoremstyle{definition}
\newtheorem{definition}[theorem]{Definition}

\theoremstyle{remark}
\newtheorem{remark}[theorem]{Remark}
\numberwithin{equation}{section}

\DeclareMathOperator{\poly}{poly}
\DeclareMathOperator{\negl}{negl}

\newcommand{\ketbra}[2]{|#1\rangle\langle#2|}

\DeclareMathOperator{\rank}{rank}
\DeclareMathOperator{\spn}{span}

\DeclareMathOperator{\Tr}{Tr}

\DeclareMathOperator{\id}{id}

\newcommand{\eps}{\varepsilon}

\newcommand{\bbC}{\mathbb{C}}
\newcommand{\bbE}{\mathbb{E}}

\newcommand{\CA}{\mathcal{A}}

\makeatletter
\renewcommand{\paragraph}{%
  \@startsection{paragraph}{4}%
  {\z@}{2.25ex \@plus 1ex \@minus .2ex}{-1em}%
  {\normalfont\normalsize\bfseries}%
}
\makeatother
\title{Random-Oracle Unitary Synthesis is Impossible}

\author[1]{\quad\quad Andrew Huang}
\author[2]{Akshar Ramkumar}
\author[3]{John Wright}
\affil[1]{University of California, Berkeley}
\affil[2]{University of Cambridge}
\affil[3]{University of California, Berkeley}

\begin{document}

\pagenumbering{gobble}

\clearpage\maketitle
\thispagestyle{empty}

\begin{abstract}
A major open problem in quantum complexity is the question of unitary synthesis---namely, is it possible to efficiently implement any unitary, given the ability to evaluate any classical function in superposition?
Inspired by a recent work by Brakerski and Yuen (CRYPTO 2026), we propose an ``average-case'' unitary synthesis problem, which asks whether it is possible to efficiently implement a Haar random unitary given access to a random Boolean function.
We then demonstrate a superpolynomial query lower bound, establishing that unitary synthesis in this model is impossible.

Along with our main result, we settle in the negative a question Brakerski and Yuen pose, of whether scalable pseudorandom unitaries (PRUs) can be implemented in the ROM-PRU model.
In particular, we show that in the ROM-PRU model, no $\Omega(N^{1+\gamma})$ unitary design on $n$ qubits, where $N=2^n$, can be efficiently achieved for any $\gamma > 0$.
On the other hand, we construct an $O(N)$-design in $\poly(n)$ queries in the same model, surpassing the previously best known results which constructed $O(\sqrt{N})$-designs.
\end{abstract}

\newpage
\tableofcontents

\newpage
\pagenumbering{arabic} 
\setcounter{page}{1}   

\section{Introduction}\label{sec:introduction}

Although all unitary transformations are allowed by the laws of nature, only a tiny sliver of them can actually be implemented efficiently on a quantum computer. 
This distinction, between $n$-qubit unitaries $U$ which are easy to implement and those which are hard, is at the heart of many of the most fundamental questions in quantum computing, as there are many important unitaries for which we still cannot determine which of the two cases they fall into.
A similar question has been deeply studied in the world of classical computing, where it is just as unclear how to determine precisely which Boolean functions $f:\{0, 1\}^n \rightarrow \{0, 1\}$ are hard to compute. It is therefore a matter of great interest to understand whether there is a rigorous connection between the hardness of \textit{unitaries} and \textit{functions}, or whether the similarity is merely superficial.

This question, known as the \emph{Unitary Synthesis Problem}, was first posed by Aaronson and Kuperberg~\cite{AK07}, who asked whether for every unitary $U$ there exists a Boolean function $f$ such that we can efficiently compute $U$ given an oracle to $f$.
If true, it would mean that behind every hard-to-compute unitary, there is an equally hard-to-compute function;
if false, it would mean some unitaries remain hard, even in a world where all functions are easy.
Answering this question has become especially urgent in recent years due to the advent of quantum cryptography, where the underlying security is based on the hardness of computing certain unitaries $U$~\cite{JLS18,GJMZ23}.
Without a resolution to the Unitary Synthesis Problem, it is unclear whether traditional complexity theory (which is about the hardness of computing Boolean functions) has anything to say about such questions, or whether we will need to invent a new, intrinsically quantum, complexity theory~\cite{BEM+23,Zha25}.

Progress on the Unitary Synthesis Problem has so far been rather modest.
The best evidence against it comes from a lower bound that states that one cannot in general efficiently synthesize $U$ with just a single query to the function $f$~\cite{LMW24}.
However, it was recently shown that the proof techniques behind this result can be extended to apply to unitaries $U$ which can be synthesized with two (!) queries~\cite{DLM26}, so it is not clear how strong a piece of evidence this is.

On the flip side, designing algorithms for unitary synthesis has also proven challenging, and what we have so far is a motley crew of upper bounds for different parameter regimes.
There is, for example, a simple procedure for unitary synthesis based on the Bernstein-Vazirani algorithm~\cite{BV93} which only makes one query to the function $f$, but it is inefficient as it requires exponentially many ancilla qubits~\cite{LMW24}.
On the other hand, the best known result for a polynomial space algorithm, due to Rosenthal~\cite{Ros26}, makes exponentially many \emph{queries} to the function $f$.
Our computationally efficient unitary synthesis algorithms are even more limited: they are few in number, and they can only synthesize small, highly structured classes of unitaries (see \cite{Yue22,Yue22b} for a good survey of these), none of which resemble general unitaries.

This was at least the case until the recent remarkable construction of \emph{pseudorandom unitaries (PRUs)} due to~\cite{MPSY24,MH25}.
A PRU is a family of efficiently-computable $n$-qubit unitaries $\{U_k\}_k$ such that a randomly chosen $U_k$ is indistinguishable from a Haar-random unitary to query-bounded adversaries.
The particular family that the work~\cite{MPSY24} introduced is known as the \textsf{PFC} ensemble, and a random unitary in this ensemble can be computed efficiently in the \emph{Random Oracle Model (ROM)}, in which one is given oracle access to a uniformly random function $f$.
The breakthrough work of Ma and Huang~\cite{MH25} established that a random $U_k$ drawn from this ensemble is indistinguishable from a Haar-random unitary to arbitrary algorithms that make up to $o(\sqrt{N})$ queries (where $N = 2^n$).
This gave the first evidence that polynomial-time quantum algorithms can exploit access to a function $f$ to synthesize a large class of unitaries which, while not completely generic, still possess many features of Haar-random unitaries.

But can we do \emph{better}?
In previous constructions, the query-security of a pseudorandom unitary depended upon the number of qubits that the unitary acts on. 
What cryptographers would really like is to decouple the security of the PRU from the number of qubits involved. 
This notion, known as a \textit{scalable} PRU, would be a PRU which is parameterized by $\kappa$, the desired security parameter.\footnote{For a given choice of $\kappa$, this means that the PRU is secure against $2^\kappa$ queries. The computational complexity of the algorithm that synthesizes the PRU is also allowed to scale with $\kappa$, ideally with a polynomial dependence.}
While the \textsf{PFC} ensemble is known to be distinguishable from a Haar-random unitary with $O(\sqrt{N})$ queries, by choosing $\kappa$ appropriately, a scalable PRU can achieve any $\poly(N)$ security. 

This question was considered in the recent work of Brakerski and Yuen~\cite{BY26}, who showed that designing sufficiently good scalable PRUs in the random oracle model would resolve the Unitary Synthesis Problem in the \emph{positive} direction!
In fact, they showed something even stronger: once you have designed a PRU in the random oracle model secure against $O(N^2)$ queries, you must have solved the Unitary Synthesis Problem.
The reason is simple: once you are allowed to make $O(N^2)$ queries, you can run the unitary tomography algorithm of~\cite{HKOT23} and learn the unitary up to some approximation error.
So the only way a PRU can be secure against $O(N^2)$ queries is if it closely mimics a Haar-random unitary, which immediately yields a circuit that can approximately synthesize most unitaries.

This result suggests two equally tantalizing possibilities.
On the one hand, if one believes that efficient unitary synthesis is not possible in general, then it may suggest that scalable PRUs do not exist.
This would rule out the only existing approach to construct what is a very natural cryptographic object.
On the other hand, if one believes that efficient unitary synthesis \emph{is} possible, then this result provides a clear route for achieving it: ``simply'' take the PFC construction and modify it to improve its security by a mere polynomial factor, from $o(\sqrt{N})$ queries to $O(N^2)$.

So which is it? This is the question we study in this work.

\subsection{Our Results}

To state our results, let us begin by highlighting a key difference between many unitary synthesis algorithms and those that arise from PRU constructions in the random oracle model.
As we have already seen, to compute a unitary $U$, a unitary synthesis algorithm is given query access to a Boolean function $f_U$ to help it, and this function may be highly structured and may depend on the unitary $U$ in arbitrary and complicated ways.
For example, the Bernstein-Vazirani-based unitary synthesis algorithm takes an $N \times N$ unitary $U$, represents it as a bitstring of length $\widetilde{O}(N^2)$, takes the $2^{\widetilde{O}(N^2)}$-bit-long Hadamard encoding of this string, and uses this as the highly structured function $f_U$.
On the other hand, a PRU algorithm in the random oracle model, such as the algorithm which synthesizes $\mathsf{PFC}$ unitaries, takes in a uniformly random oracle $f$, which has no structure to speak of, and is still expected to convert it to a (somewhat) random-looking unitary.
Motivated by this, we introduce the definition of a \emph{random-oracle unitary synthesis algorithm}, which is one that converts random bits into Haar random unitaries.

\begin{definition}[Random-oracle unitary synthesis, informal; see Definition~\ref{def:rom_synthesis}]
    Let $A^{(\cdot)}$ be an oracle algorithm with $n$ qubits of input and $n$ qubits of output, and let $A^{(f)}$ be the channel it computes given oracle access to the Boolean function $f$.
    Then $A^{(\cdot)}$ is a \emph{random-oracle unitary synthesis algorithm with error $\epsilon$} if there is a coupling $(f, U)$ between a uniformly random Boolean function $f$ and a Haar random unitary $U \in U(N)$ such that
    \begin{equation*}
        \bbE_{(f, U)} D_{\diamond}(A^{(f)}, \Phi_U) \leq \epsilon,
    \end{equation*}
    where $\Phi_U$ is the unitary channel corresponding to $U$.
\end{definition}

This definition is meant to capture the type of unitary synthesis algorithm we could hope for if it was possible to extend the $\mathsf{PFC}$ ensemble to a full unitary synthesis algorithm.
It is inspired by the definition of ROM-PRUs due to Brakerski and Yuen~\cite{BY26}, and the two definitions can be found formally stated in Section~\ref{sec:rom_synthesis}, where we compare them.
Our first main result is that there is no efficient algorithm for random-oracle unitary synthesis, so long as the error $\epsilon$ is even slightly nontrivial (meaning that it is bounded away from 1 by a $1/\poly(n)$ factor).

\begin{theorem}[No random-oracle unitary synthesis]\label{thm:no-ROM-unitary}
    Suppose $A^{(\cdot)}$ is an $n$-qubit random-oracle unitary synthesis algorithm with error $\epsilon$, and $\epsilon \leq 1 - 1/\poly(n)$. Then if $k$ is the number of queries $A^{(\cdot)}$ makes to the random oracle, we must have
    \begin{equation*}
        k \geq \frac{\sqrt{N}}{\poly(n)}.
    \end{equation*}
\end{theorem}

Hence, any efficient unitary synthesis algorithm must use Boolean functions which are highly structured.
Via the results of Brakerski and Yuen, this also tells us, intuitively, that although there are PRUs in the random oracle model secure against adversaries which make $o(\sqrt{N})$ queries, they cannot be secure against adversaries which make $O(N^2)$ queries.
For our next results, we will identify the number of queries between these two extremes at which PRUs must become insecure.

First, why are all of our PRUs ``stuck'' at being secure against only $o(\sqrt{N})$ queries?
For the $\mathsf{PFC}$ ensemble, the answer is simply that it is \emph{insecure} against any larger number of queries.
But for other candidate PRUs, the answer is in the dominant paradigm for analyzing PRUs, which is to focus on the so-called ``distinct subspace''. 
The distinct subspace arises naturally if we imagine a $T$-query adversary that prepares $T$ standard basis vectors $\ket{x_1} \ldots \ket{x_T}$ in tensor product and then queries the unknown unitary in parallel across these states.
In this case, the distinct subspace corresponds to the tensor products where the elements $x_1, \ldots, x_T \in [N]$ are all distinct.
When the number of queries $T = o(\sqrt{N})$, there are often tricks we can play to pretend as though we are in the distinct subspace, and doing so can make the remaining analysis significantly easier.
However, once $T = \Omega(\sqrt{N})$, one is forced to consider non-distinct subspaces as well, and we have not had any luck in analyzing PRUs in this regime.
It has therefore been an open question if the reason we are stuck at $o(\sqrt{N})$-query security is merely technical, or if something more fundamental is going on.
We show that it is merely a technical obstruction after all, and that one can push the security all the way to $T = O(N)$ queries.

\begin{theorem}[An $O(N)$-query secure PRU]\label{thm:our-pru}
    There is a PRU in the random oracle model which is secure against any adversary which makes $O(N)$ queries.
\end{theorem}

We note that the PRU in this construction is efficient in the sense that the number of times it needs to query the random oracle is only $\poly(n)$; however, unlike in the case of the $\mathsf{PFC}$ ensemble, we do not know how to implement it in a computationally efficient manner. We do believe that this is only a technical barrier, however, and that $O(N)$-query security should be achievable even by PRU algorithms which are query-efficient and computationally efficient.

The PRU-synthesizing algorithm in this construction is easy to describe: it computes the unitary
\begin{equation*}
    U = F_n H_n \cdots F_1 H_1 F_0,
\end{equation*}
where each $H_i$ is a fixed $n$-qubit unitary, and each $F_i$ can be viewed as a query to a separate $n$-bit random function $f_i$.  
This unitary $U$ is generated from $\widetilde{O}(N)$ random bits.
To push beyond $\widetilde{O}(N)$-query security, we would need our PRU algorithm to use more than $\widetilde{O}(N)$ random bits, and the only way to do so with $\poly(n)$ queries is to include additional ancilla qubits in the oracle's input register.
But performing queries on the ancilla qubits is fraught with danger, as it is easy to do so in a way which creates entanglement between the first $n$ qubits and the ancilla qubits, which would mean we are no longer computing a unitary on the first $n$ qubits, as we are intending to.
This ancilla issue has proven to be a significant barrier in designing even general unitary synthesis algorithms, but it seems especially challenging in the random oracle model, where we have no control over the form the oracle can take.

Our final result shows that this issue is not just a technical barrier: no query-efficient PRU in the random oracle model can achieve security beyond $\widetilde{O}(N)$ queries.

\begin{theorem}[No query-efficient ROM-PRUs beyond $\widetilde{O}(N)$-query security]\label{thm:no-ROM-PRU}
    Fix a constant $\gamma>0$. Suppose there is an $n$-qubit PRU in the random oracle model which is secure against $T = \Omega(N^{1+\gamma})$ queries. Then if $k$ is the number of queries to the random oracle needed to synthesize the PRU, we must have
    \begin{equation*}
        k \geq \frac{N^{\min\{\gamma, 1\}/2}}{\poly(n)}.
    \end{equation*}
\end{theorem}

This result smoothly interpolates between Theorems~\ref{thm:no-ROM-unitary} and~\ref{thm:our-pru}: it rules out security beyond $\widetilde{O}(N)$ queries for query-efficient PRUs in the random oracle model. Thus, the query security achieved by the construction in Theorem~\ref{thm:our-pru} is optimal up to factors polynomial in $n$.
And for PRUs which are $T = \Omega(N^2)$-query secure, it shows that the synthesizing algorithm must use roughly $N^{1/2}$ queries, matching the bound in Theorem~\ref{thm:no-ROM-unitary}.

\subsection{Technical Overview}

\subsubsection{Lower Bounds}

\paragraph{Modeling the algorithm.}
To discuss our lower bounds, let us begin with how we mathematically model unitary synthesis algorithms in the random oracle model.
An algorithm which synthesizes $n$-qubit unitaries using $m$ total qubits can be represented as a quantum circuit of the form
\begin{equation}\label{eq:the-form}
W_f = V_k \mathcal{O}_f \dots \mathcal{O}_f V_0.
\end{equation}
Here, $V_0, \dots, V_k \in U(2^m)$ are fixed unitaries on $m$ qubits.
In addition, $\mathcal{O}_f$ is the diagonal phase unitary corresponding to a random function $f:\{0, 1\}^{m_0} \to \{0, 1\}$ for $m_0 \leq m$, which acts as
\begin{equation*}
\mathcal{O}_f\ket{x0^{m - m_0}} = (-1)^{f(x)}\ket{x0^{m-m_0}}.
\end{equation*}
The operation of the algorithm is therefore to implement the following $n$-qubit channel: on input an $n$ qubit state $\ket{\psi}$, (i) append $m-n$ ancilla qubits initialized to $\ket{0}$, (ii) apply the circuit, resulting in $W_f \ket{\psi 0^{m-n}}$, (iii) trace out the $m-n$ ancilla qubits, and (iv) output the rest.

If $U$ is the unitary that this algorithm is meant to implement, we will assume for the time being that it does so \emph{cleanly}, meaning that it resets the ancilla qubits to $\ket{0}$ prior to them being traced out;
in other words,
\begin{equation*}
W_f \ket{\psi 0^{m-n}} \approx (U \ket{\psi}) \ket{0^{m-n}}.
\end{equation*}
Equivalently, if $J$ is the isometry that takes $\ket{\psi}$ on $n$ qubits to $\ket{\psi0^{m - n}}$, then the algorithm cleanly implements $U$ if
\begin{equation*}
A_f \coloneqq J^\dag W_f J \approx U. 
\end{equation*}
Assuming this will help simplify our analysis, but we will see how to remove this assumption later.

\paragraph{The case of one query.}
Our goal is to prove a lower bound on the number of queries needed to obtain ROM-PRUs or a random-oracle unitary synthesis algorithm.
To do so, we need to constrain the number of unitaries that can be synthesized by a circuit of the form given in Equation~\eqref{eq:the-form}.
At a high level, we will do this by arguing that even though the algorithm takes in a large number of random bits from the random oracle---$2^{m_0}$ of them, in fact---only a small number of these bits actually impact the final unitary that the algorithm computes. 
If this is the case, then we can argue that the number of random bits the algorithm depends on is actually significantly smaller than $N^2$, in which case it cannot be simulating a Haar random unitary.

To illustrate the strategy we take to accomplish this, let us first consider the simplest case, in which the circuit only makes one query to the random oracle $f$.
In this case, $A_f$ can be written as
\begin{equation*}
A_f = J^\dag V_1 \mathcal{O}_f V_0 J. 
\end{equation*}
Now, for a fixed $f$, we would like to understand the extent to which changing its value at a given location $x \in \{0, 1\}^{m_0}$ affects the unitary $A_f$ that the algorithm computes. 
If $f^{\oplus x }$ denotes the function $f$ with its value at $x$ flipped, then
\begin{equation*}
A_{f^{\oplus x}} - A_f = \pm 2 J^\dag V_1 \Pi_x V_0 J.
\end{equation*}
We can therefore bound the Hilbert-Schmidt distance between these two cases as
\begin{align}
\lVert A_{f^{\oplus x }} - A_f \rVert_F^2
= 4 \lVert J^\dag V_1 \Pi_x V_0 J \rVert_F^2
\leq 4 \lVert J^\dag V_1 \Pi_x\rVert^2_F \lVert \Pi_x V_0 J\rVert^2_F.
\end{align}
Motivated by this, let us define a vector of weights $w$ via
\begin{equation}
w(x) = \lVert J^\dag V_1 \Pi_x\rVert^2_F + \lVert \Pi_x V_0 J\rVert^2_F,
\end{equation}
which we will interpret as representing how much one can affect the unitary $A_f$ by changing $f(x)$.
This  vector is very similar to the ``weight vector'' from \cite{LMW24}.\footnote{Indeed, if the second term in the vector is appropriately normalized, then it is identical to the weight vector used in \cite{LMW24}. The first term is unnecessary in their work because in the single query case, it can be absorbed into a final measurement. While both works use the weight vector to understand the influence of different inputs, we do not yet have a formal connection between the two proof strategies.}

We would like these weights to satisfy two properties.
First, it should be the case that on a given input $f$, very few coordinates have high weight.
In particular, let us consider how many bitstrings $x$ satisfy a certain weight threshold
\begin{equation*}
w(x) > \tau,
\end{equation*}
where we will typically imagine $\tau$ to be some sufficiently small constant. 
Note that we can bound the sum of all weights via
\begin{align*}
\sum_{x \in \{0, 1\}^{m_0}} w(x) &= \left(\sum_x \lVert J^\dag V_1 \Pi_x \rVert_F^2\right) + \left(\sum_x \lVert \Pi_x V_0 J \rVert_F^2\right)\\
&= \lVert J^\dag V_1 \rVert_F^2+ \lVert V_0 J \rVert_F^2\\
&= 2N,
\end{align*}
since $V_0 J$ and $(J^\dag V_1)^\dag$ are both isometries with input dimension $N$. 
So, while there are $M_0 = 2^{m_0}$ choices for $x$, only $O(N)$ of them can satisfy $w(x) > \tau$ when $\tau$ is a constant.
That is to say, nearly all inputs have low weight.

Second, it should be the case that on a given input $f$, flipping the value of $f$ on one or more coordinates with low weight should not change the unitary computed by $A_f$ by a significant amount.
Starting with one coordinate at a time, for any $x$ such that $w(x) \leq \tau$,
\begin{equation*}
\lVert A_{f^{\oplus x}} - A_f \rVert_F^2
\leq 4 \lVert J^\dag V_1 \Pi_x\rVert^2_F \lVert \Pi_x V_0 J\rVert^2_F
\leq 4 w(x)^2
\leq 4\tau w(x).
\end{equation*}
Generalizing this, suppose that $f$ and $f'$ only differ at inputs $x \in S \coloneqq \{x_1, \ldots, x_{|S|}\}$ that satisfy the weight threshold. Letting $f^{\oplus T}$ denote the function which is equal to $f$ outside of $T$ and has its bits flipped on coordinates in $T$, we see that 
\begin{align}
\lVert A_{f'} - A_f \rVert_F^2 &\leq |S| \sum_{i = 1}^{|S|} \lVert A_{f^{\oplus \{x_1, \ldots, x_i\}}} - A_{f^{\oplus \{x_1, \ldots, x_{i-1}\}}} \rVert_F^2 \tag{using $\lVert \sum_{i=1}^L C_i \rVert_F^2 \leq L\sum_{i = 1}^L \lVert C_i \rVert_F^2$}\\
&\leq 4|S|\tau \sum_{i = 1}^{|S|} w(x_i)\nonumber\\
&\leq 8\tau M_0 N.  \tag{because $|S| \leq M_0$}
\end{align}
If the $|S|$ were \emph{not} in the expression above, i.e.\ if the upper bound were simply $8\tau N$, then it would be a strong upper bound.
This would mean that every $A_{f'}$ reachable from $A_f$ by flipping low-weight coordinates is at most $\sqrt{8 \tau N}$ in Hilbert-Schmidt distance away from $A_f$.
For $\tau$ sufficiently small, this would place every such $A_{f'}$ in a nontrivial ball around $A_f$, since the Hilbert-Schmidt diameter of the unitary group is $O(\sqrt{N})$.
In doing so, it would show that there are only $O(N)$ inputs of $f$ that can change $A_f$ significantly---however, $O(N)$ degrees of freedom is not sufficient to cover all unitaries, or even unitary designs for $T=N^{1+\gamma}$. 

However, $|S|$ \emph{is} in the expression above, and its presence expresses the possibility that altering the value of $f$ at different inputs could lead to matrix differences
\begin{equation*}
     A_{f^{\oplus \{x_1, \ldots, x_i\}}} - A_{f^{\oplus \{x_1, \ldots, x_{i-1}\}}}
\end{equation*}
which constructively interfere. 
Instead, what we would like to show is that these differences behave as though they are pairwise orthogonal, in which case we would not pick up the additional factor of $|S|$.

\paragraph{Proving orthogonality.}
How do we ensure that these matrix distances are orthogonal? We solve the problem by considering our norm
\begin{equation*}
    \lVert A_{f'} - A_f \rVert_F^2
\end{equation*}
\emph{in expectation over }$f$ and $f'$. 
We consider a uniformly random $f$, and resample its value $f(x)$ iteratively at every single choice of $x \in \{0, 1\}^{m_0}$ which has low weight. 
The process of resampling a uniformly random $f$ at each index sequentially is a (time-inhomogeneous) symmetric Markov chain. 
It is symmetric because at each step, the probability that $f$ transitions to $f'$ and the probability that $f'$ transitions to $f$ are equal. 
If $W$ is the number of low-weight coordinates in $f$, consider the sequence $A_{f}, A_{f_1}, \ldots, A_{f_{W}}$ of blocks corresponding to the iterative process that resamples $f$ on each coordinate, one at a time.
There are known techniques \cite{NPSS06} that can use the symmetry of the Markov chain to decompose this sequence into a ``forward martingale'' and a ``backward martingale'', and the operators in each sequence can be shown to be pairwise-orthogonal in expectation. 
This orthogonality is what allows us to complete the result. 

How does it work? Instead of considering each term
\begin{equation}
A_{f_i} - A_{f_{i-1}},
\end{equation}
we consider 
\begin{equation}
(A_{f_i} - \overline{A}_{i}) + (\overline{A}_{i} - A_{f_{i-1}}). 
\end{equation}
Here, $\overline{A}_{i}$ is the average of the possible values of $A_{f_i}$, given $f_{i-1}$. 
Now, notice that for any $f_j$ before $f_i$, $\mathbb{E}[A_{f_i} - \overline{A}_i | f_j] = 0$. 
As a consequence, $\left(A_{f_i} - \overline{A}_i\right)_i$ form a \textit{martingale difference sequence}, which immediately implies that 
\begin{equation}
\mathbb{E}\left[\left(A_{f_i} - \overline{A}_i\right)\left(A_{f_j} - \overline{A}_j\right)\right] = 0.
\end{equation}
Hence, this first sequence of operators is orthogonal in expectation.
On the other hand, $\overline{A}_i - A_{f_{i-1}}$ is a \textit{also} a martingale difference sequence, although it is less obvious to see why.
Because our Markov chain is symmetric, we can consider $A_{f}, \dots, A_{f_{W}}$ backwards, and it is still a Markov chain. 
After doing so, $\overline{A}_i - A_{f_{i-1}}$ is completely analogous to $A_{f_i} - \overline{A}_i$, and orthogonality follows for the second sequence as a consequence. 

Hence, taking an expectation over $f$ resolves the orthogonality issue---our original sum decomposes as the sum of two orthogonal sequences, which is enough to obtain the sufficient bound
\begin{equation*}
\mathbb{E}\left[\lVert A_{f'} - A_f \rVert_F^2 \right]\leq 2 \mathbb{E}\left[ \sum_{x \in S} \lVert A_{f^{\oplus x}} - A_f \rVert_F^2\right].
\end{equation*}
If we did not take the expectation, our argument would have implied worst-case unitary synthesis is not possible. Thus, in taking the expectation, we sacrifice the possibility of a worst-case result, but obtain a strong statement about a uniformly random $f$.  

\paragraph{The case of multiple queries.}
All of these arguments work for one query, but what happens if there are more? 
In this case, there are two important qualitative changes. 
First, if the algorithm makes $k$ queries, then we must now express $A_{f^{\oplus x}} - A_f$ as a sum of $k$ hybrids:
\begin{equation*}
A_{f^{\oplus x}} - A_f = \sum_{i =1}^kH_i.
\end{equation*}
In each hybrid, one of the oracle queries, $\mathcal{O}_f$, is replaced by $\mathcal{O}_{f'}$. 
This is a relatively minor change, although the hybrid introduces an extra factor of $k^2$, which is where the query-dependence of our theorems arises from. 
We also must define our weight function differently now. 
If $L_{f, i} =  J^\dag V_k \mathcal{O}_f \dots V_i$ and $R_{f, i} =  V_{i-1} \dots V_1 \mathcal{O}_f V_0 J$, then we define
\begin{equation}
w_f (x) = \sum_{i = 1}^k  \left(\lVert L_{f, i }\Pi_x \rVert_F^2 + \lVert \Pi_x R_{f, i} \rVert_F^2\right).
\end{equation}
Now, the weight function is a sum of the influence of every input $x \in \{0, 1\}^{m_0}$, for each of the queries. 
But notably, the weight function now \emph{depends on $f$}. 
This is the second qualitative difference, and it forces the resampling algorithm for $f$ to be adaptive. 
Indeed, whenever an output of $f$ is resampled, the weights of other inputs change. 
Despite this complication, all the subsequent analysis survives, and we can still obtain the result that resampling most inputs of $f$ only perturbs $A_f$. 

Recall that all the arguments described in this section applied only to clean synthesis.
While they are strong enough to prove strong lower bounds against both unitary designs and random-oracle unitary synthesis in this setting, we need to consider the possibility that the unitary is implemented while introducing garbage into the ancilla.
If the channel corresponding to $W_f$ has small diamond distance $\epsilon \approx 0$ from the channel corresponding to a unitary $U \in U(N)$, we can consider the circuit
\begin{equation*}
W_f^\dag (P \otimes I) W_f,
\end{equation*}
for any fixed Pauli $P \in \mathcal{P}_n$. 
This is a rotation of the Pauli operator $P$, but $W_f^\dag$ erases the garbage ancilla that $W_f$ introduced.
This circuit is therefore close to a clean implementation of $U^\dag P U$, and hence a low error dirty algorithm can implement a very large class of reflections, which is sufficient to apply the clean lower bound described above. 
However, this argument fails at a diamond distance of $\frac{1}{2}$, since there are two copies of $W_f$ included, which can blow up the error by a factor of two. 
We are able to handle this extra technicality by adapting the argument, instead considering the expectation over a random Pauli to overcome the obstacle and extend our impossibility to the full dirty regime. 

\subsubsection{Upper Bound}
Recall that our goal is to show that the ensemble $\{F_n H_n \ldots F_1H_1F_0\}_{F_0, \ldots, F_n}$ is an $O(N)$-query secure PRU for some choice of $H_1, \ldots, H_n$. In fact, we will be able to show that with high probability over the choice of $H = (H_1, \ldots, H_n) \sim \mu_{\mathsf{Haar}}^N$, the distribution $\mu_H \sim F_n H_n \ldots F_1 H_1 F_0$ is indistinguishable from $\mu_{\mathsf{Haar}}$ to $T$-query distinguishers for $T \leq  cN$ (for some $c>0$). By an application of Markov's inequality, it will suffice to show that
    \[ \bbE_{H_1, \ldots, H_n \sim \mu_{\mathsf{Haar}}}\left[\max_{\text{$T$-query $\mathcal{A}$}}\left|\Pr_{U \sim \mu_H}[\mathcal{A}^{U, U^{\dagger}, \overline{U}, U^T} = 1] - \Pr_{U \sim \mu_{\mathsf{Haar}}}[\mathcal{A}^{U, U^{\dagger}, \overline{U}, U^T} = 1]\right|\right] \leq \negl(n). \]
For simplicity we assume that our adversaries only make forward queries to $U$, but it turns out that our proof will apply equally well to conjugate/inverse/transpose queries as well.

We begin by considering the situation for a fixed choice of $H$. We observe that the acceptance probability of any adversary $\mathcal{A}$ which makes at most $T$ queries to $U$ can be written as the linear combination of entries of the matrices $\{U^{\otimes q} \otimes \overline{U}^{\otimes q}\}_{0 \leq q \leq T}$. A well-known fact in representation theory known as \emph{mixed Schur-Weyl duality} therefore implies that the acceptance probability of $\mathcal{A}$ can actually be expressed in terms of the irreducible representations (irreps) $\{\rho_{\lambda, \nu}\}_{\lambda, \nu}$ of the maps
    \[ U \mapsto U^{\otimes q} \otimes \overline{U}^{\otimes q}, \quad 0 \leq q \leq T \]
which are indexed by partitions $\lambda, \nu \vdash r$ where $0 \leq r \leq T$.

Our high-level proof strategy will therefore be to show that our ensemble $\mu_H$ ``fools" all such ``low-degree" (nontrivial) irreps. Roughly speaking, our goal will be to bound $\bbE_{F_0, \ldots, F_n}[\rho_{\lambda, \nu}(F_nH_n \ldots F_1H_1F_0)]$ for each irrep $\rho_{\lambda, \nu}$. By the independence of $F_0, \ldots, F_n$, we can factor our expression as follows:
    \[ \bbE_{F_0, \ldots, F_n}[\rho_{\lambda, \nu}(F_nH_n \ldots F_1H_1F_0)] = \bbE_{F_n}[\rho_{\lambda, \nu}(F_n)]\rho_{\lambda, \nu}(H_n) \ldots \bbE_{F_1}[\rho_{\lambda, \nu}(F_1)]\rho_{\lambda, \nu}(H_1)\bbE_{F_0}[\rho_{\lambda, \nu}(F_0)]. \]
Let us examine the term $\bbE_{F_i}[\rho_{\lambda, \nu}(F_i)]$. Since $F_i$ is drawn randomly from the set of all possible phase functions (which form a group), it turns out that $\bbE_{F_i}[\rho_{\lambda, \nu}(F_i)]$ is actually equal to a \emph{projection}, which we denote by $P_{\lambda, \nu}$. If $F_i$ were actually a Haar-random unitary, then this projection would act as the zero projection for all nontrivial representations, which reflects the fact that a Haar-random unitary is a secure PRU under any number of queries. Our hope would be that $P_{\lambda, \nu}$ is close to zero in the sense that it is relatively low-rank for all irreps $\rho_{\lambda, \nu}$. 

Unfortunately, although one can show that $P_{\lambda, \nu}$ projects onto an exponentially small part of the overall Hilbert space, there are also an exponential number of irreps! And indeed, it must be the case that this bound cannot suffice, as a single random phase $F_i$ is trivial to distinguish from a Haar-random unitary.

Likewise, if we just iteratively applied more random phases, this would not change our ensemble since applying the projection $P_{\lambda, \nu}$ a second time does not have any effect. This leads us to introduce the interleaving unitaries $(H_1, \ldots, H_n)$, so that we can leverage their \emph{mixing properties}. Intuitively, since $H_1, \ldots, H_n$ are generic, they effectively scramble any state in the image of $P_{\lambda, \nu}$, so that each time we apply it, all but an exponentially small part of the adversary's state gets killed by the projection. By compounding this effect $O(n)$ times, we obtain an exponentially improved bound, which is sufficient to compensate for the exponential number of irreps.

Due to the fact that for all choices of $(H_1, \ldots, H_n)$, the ensemble $\{F_n H_n \ldots F_1H_1F_0\}_{F_0, \ldots, F_n}$ is invariant under global (complex) phases, the reduction of \cite{TW25} immediately yields security against controlled queries as well.

Crucially, although each instantiation of our ensemble does \emph{not} use any randomness in the choice of $(H_1, \ldots, H_n)$, we are only able to show that there exists a choice (and in fact, most choices) of $(H_1, \ldots, H_n)$ that makes $\{F_n H_n \ldots F_1H_1F_0\}_{F_0, \ldots, F_n}$ an $O(N)$-secure PRU, which is why our construction is only query-efficient but not computationally efficient. Preliminary exploration has indicated that our general approach can be leveraged for an appropriate choice of efficiently implementable (deterministic) unitaries $(H_1, \ldots, H_n)$,\footnote{One can for example take $H_1 = \ldots = H_n$ to be the $n$-qubit Hadamard.} to obtain an $O(\sqrt{N})$-secure \emph{strong} pseudorandom unitary, which would be a significant improvement over the bounds shown for existing constructions of strong PRUs \cite{MH25}.\footnote{In particular, $O(\sqrt{N})$ security has only been shown for forward PRU constructions, and existing techniques only give the much worse bound of $O(N^{1/16})$ strong security.} However, since we have not verified all such proofs, we have chosen not to include the results here.

To our knowledge, this is the first time that the idea of bounding the bias of the unitary distribution at each irreducible representation has been used to show the security of unitary designs. Unlike previous proof techniques for analyzing designs (e.g. the path recording oracle of \cite{MH25, FLM+26}), our proof easily handles the strong query case as it does not distinguish between query types. Finally, we note that the idea of finding constructions that ``fool low-degree irreducible representations" bears some light resemblance to the construction of epsilon-biased sets for nonabelian groups \cite{CMR13}, although the goal there is to bound the maximum bias over all irreducible representations (rather than their sum). 

\subsection{AI Use Disclosure}
In this section, we will describe the precise work contributed by models, and contributed by human authors. 

We (the human authors) were inspired by the connection between PRUs in the random-oracle model and unitary synthesis described in \cite{BY26}.
We believed that formulating and attacking unitary synthesis in an average-case setting (with random oracles in place of an arbitrary oracle) would be natural, and could be a stepping stone to impossibility results for unitary synthesis and ROM-PRUs. 
We created a usable definition of random-oracle unitary synthesis and deployed GPT 5.6 Sol for several hours. 
It was able to solve the problem and an additional prompt obtained the ROM-PRU lower bound as well. 
We then engaged in interactive sessions with GPT 6 Astra to understand the proof and simplify its structure.

Concurrently, we were working on constructing efficient unitary designs to surpass the $\sqrt{N}$-query barrier in the random-oracle model. 
GPT 5.6 Sol was able to prove the $O(N)$-query security of a target construction that the authors proposed, for a \textit{query-efficient} unitary design. 
This construction provides a matching upper bound to our lower bound against ROM-PRUs.

After digesting and organizing the result, we produced a draft, which was then verified with AI assistance.
We take full accountability for the correctness of the results in this work. 

\paragraph{Organization.}
In Section \ref{sec:prelims}, we establish the preliminaries for our work, including standard notation and the essential background in representation theory.
In Section \ref{sec:rom_synthesis}, we recall the definitions of unitary synthesis and pseudorandom unitaries before proposing our new definition of random-oracle unitary synthesis. 
In Section \ref{sec:lower_bound}, we prove the impossibility of efficient random-oracle unitary synthesis and (forward non-adaptive) pseudorandom unitaries in the random-oracle model with security against more than $\widetilde{O}(N)$ queries. 
In Section \ref{sec:linear_design}, we construct a strong pseudorandom unitary in the random-oracle model secure against $O(N)$ queries. 

\section{Preliminaries} \label{sec:prelims}
\subsection{Notation}
For an $n$ qubit system, we sometimes denote the dimension $N = 2^n$. In general, we use capitalization to denote the dimension of a system of qubits, including $m$ and $M$ as well as $m_0$ and $M_0$. 
We denote the set of $n$-qubit Pauli operators by
\[
\mathcal P_n
=\{P_1\otimes\cdots\otimes P_n:
P_j\in\{I,X,Y,Z\}\}.
\]
In particular, $|\mathcal P_n|=4^n$, and every element is Hermitian and unitary.
For $n\leq m$, let $J$ be the isometry defined by $J\ket{\psi}=\ket{\psi}\otimes\ket{0^{m-n}}$.
For a unitary $U$, write $\Phi_U$ for the channel $\Phi_U(\rho)=U\rho U^\dag$.
We write $i(\rho)=J\rho J^\dag$ for the channel that appends $m-n$ zero ancillas, while we denote the partial trace over the last $m-n$ qubits by $\pi$. 
Thus, for an $m$-qubit unitary $W$, the channel $\pi\circ\Phi_W\circ i$ appends zero ancillas, applies $W$, and discards the ancillas.

We will often consider a function $f:\{0,1\}^{m_0}\to\{0,1\}$, for $m_0\le m$ where $m$ is the total number of qubits. Then, 
define the phase oracle $\mathcal O_f$ acting on
$m$ qubits by
    \[ \mathcal O_f = \sum_x (-1)^{f(x)}\ketbra{x}{x} \otimes I_{m-m_0}. \] 

For a matrix $A$, we write $\lVert A\rVert$ for its
operator norm, equal to its largest singular value.
Its trace norm and Frobenius norm are, respectively,
\[
\lVert A\rVert_1
=\text{tr}\sqrt{A^\dag A},
\qquad \text{and} \qquad
\lVert A\rVert_F
=\sqrt{\text{tr}(A^\dag A)}.
\]
For density operators $\rho,\sigma$ and quantum
channels $\Phi,\Psi$, we use the trace and diamond distances,
\[
D(\rho,\sigma)=\frac12\lVert\rho-\sigma\rVert_1,
\qquad
D_\diamond(\Phi,\Psi)
=\frac12\lVert\Phi-\Psi\rVert_\diamond.
\]
For channels $\Phi,\Psi$ acting on system $A$, the diamond distance is
\begin{equation*}
D_\diamond(\Phi,\Psi)=\sup_{B,\rho_{AB}}D\!\left((\Phi\otimes\id_B)(\rho_{AB}),(\Psi\otimes\id_B)(\rho_{AB})\right),
\end{equation*}
where the supremum ranges over all finite-dimensional reference systems $B$ and density operators $\rho_{AB}$ on $A\otimes B$.

\subsection{Unitary Designs}
We consider two different notions of unitary designs. 
The first we provide is a standard/forward non-adaptive design, the weakest commonly used definition and the one that we will consider for our lower bound. Our upper bound uses the strong adaptive notion defined below.
\begin{definition}[Standard non-adaptive $T$-query design]\label{def:standard_nonadaptive_design}
Let $\nu$ be a distribution over $U(N)$.
We say that $\nu$ is a \emph{$\delta$-approximate standard non-adaptive unitary $T$-design} if
\[
    \frac{1}{2}
    \left\|
        \mathcal T_{\nu}^{(T)}
        -
        \mathcal T_{\mathrm{Haar}}^{(T)}
    \right\|_\diamond
    \leq \delta,
\]
where
\[
    \mathcal T_{\mathcal D}^{(T)}(\rho)
    :=
    \mathbb E_{U\sim\mathcal D}
    \left[
        U^{\otimes T} \rho (U^\dagger)^{\otimes T}
    \right].
\]
Equivalently, no distinguisher that prepares an arbitrary joint state
on $T$ query registers and an arbitrary reference system, applies
$U^{\otimes T}$ in parallel, and then performs an arbitrary joint
measurement can distinguish $U\sim\nu$ from Haar with advantage
greater than $\delta$.
\end{definition}
We now define a strong adaptive unitary design. 
In this case, the unitary no longer has to be queried in parallel--furthermore, we allow access to the unitary, its inverse, its conjugate and transpose, as well as controlled access to each of these oracles. 
\begin{definition}[Strong adaptive design]\label{def:strong_adaptive_design}
Let $\nu$ be a distribution over $U(N)$.
For $U\in U(N)$, define the oracle family
\[
    \mathcal S_U=
    \left\{
        U,\,
        U^\dagger,\,
        \overline U,\,
        U^T,\,
        cU,\,
        cU^\dagger,\,
        c\overline U,\,
        cU^T
    \right\},
\]
where
\[
    cV    :=
    |0\rangle\!\langle 0|\otimes I
    +
    |1\rangle\!\langle 1|\otimes V.
\]

We say that $\nu$ is a $\delta$-approximate
\emph{strong adaptive unitary $T$-design} if, for every
quantum algorithm $\mathcal A$ making at most $T$ total oracle calls
to members of $\mathcal{S}_U$,
\[
\left|
    \Pr_{U\sim\nu}
        [\mathcal A^{\mathcal{S}_U}=1]
    -
    \Pr_{U\sim\mathrm{Haar}}
        [\mathcal A^{\mathcal{S}_U}=1]
\right|
\leq \delta.
\]
The oracle calls may be chosen adaptively, and arbitrary ancillary
systems and intermediate quantum operations are allowed.
\end{definition}

\subsection{Representation Theory}
We review the representation definitions and results we will need for the proof of our unitary design construction. For more details, see \cite{Wri16,Ngu23,SSW25}.

The unitary group $U(d)$ is the group of $d \times d$ unitary matrices. More generally, for $V$ a complex, finite-dimensional Hilbert space, $U(V)$ is the group of unitary operators on $V$.
\begin{definition}[Representations]
    Let $G$ be a group. A \emph{complex, unitary, finite-dimensional representation} (henceforth, a \emph{representation}) of $G$ is a pair $(\mu, V)$, where $V$ is a finite-dimensional complex Hilbert space and $\mu: G \to U(V)$ is a group homomorphism. The dimension of $V$ is called the dimension of the representation $\mu$ and is denoted $\dim(\mu)$.
\end{definition}

\begin{definition}[Characters]
    Let $(\mu, V)$ be a representation of a group $G$. The \emph{character} of $\mu$ is the map $\chi_\mu : G \to \bbC$ given by $\chi_\mu(g) := \Tr \mu(g)$.
\end{definition}

\begin{definition}[Irreducible representations]
    Let $(\mu, V)$ be a representation of a group $G$. A subspace $W \subseteq V$ is an \emph{invariant subspace} of $V$ if $\mu(g) \cdot W \subseteq W$ for all $g \in G$. An invariant subspace is \emph{trivial} if $W = \{0\}$ or $W = V$. If $V$ has a nontrivial invariant subspace, it is called \emph{reducible}, and otherwise is called \emph{irreducible}. An irreducible representation is also called an \emph{irrep}. The set of equivalence classes of irreps of $G$ will be written $\widehat{G}$.
\end{definition}
Moving forward, we will fix a representative $\widehat{\mu}_i$ from each isomorphism class of irreps and identify $\widehat{G}$ with the set $\{\widehat{\mu}_i\}$.

\begin{lemma}[Schur’s lemma]
    Let $G$ be a group, and let $(\mu, V)$ be an irreducible representation of $G$. Let $T: V \to V$ be a linear map such that $T\mu(g) = \mu(g)T$ for all $g \in G$. Then $T = \frac{\Tr T}{\dim(\mu)} \cdot I_V$.
\end{lemma}

\begin{definition}[Direct sum of representations]
    Let $G$ be a group, and let $(\mu_1, V_1), \ldots, (\mu_k, V_k)$ be representations of $G$. Then the direct sum of $\mu_1, \ldots, \mu_k$ is the representation $(\mu, V)$, where $V := \bigoplus_{i = 1}^k V_i$ and $\mu(g) := \bigoplus_{i = 1}^k \mu_i(g)$ for all $g \in G$. Each representation $\mu_i$ may appear multiple times. If $\mu_i$ appears $m_i$ times, we also write $\mu = \bigoplus_{i = 1}^k m_i \cdot \mu_i$, and refer to $m_i$ as the multiplicity of $\mu_i$.
\end{definition}

\begin{definition}[Complete reducibility]
    Let $G$ be a group, and let $(\mu, V)$ be a representation of $G$. Then $\mu$ is \emph{completely reducible} if
        \[ \mu \cong \bigoplus_{\widehat{\mu}_i \in \widehat{G}} m_i \cdot \widehat{\mu}_i \]
    for some nonnegative integers $\{m_i\}$.
\end{definition}
Finite-dimensional unitary representations are completely reducible.

\begin{definition}[Partitions]
    Let $n$ be a non-negative integer. A \emph{partition} of $n$ is a finite list $\lambda = (\lambda_1, \ldots, \lambda_m)$ of positive integers such that $\lambda_1 \geq \lambda_2 \geq \ldots \geq \lambda_m$ and $\lambda_1 + \cdots + \lambda_m = n$. We write $\lambda \vdash n$ to denote that $\lambda$ is a partition of $n$. The \emph{size} of $\lambda$ is $|\lambda| = n$, and the \emph{length} of $\lambda$, denoted $\ell(\lambda) = m$.

    Note that the empty partition $\lambda = \emptyset$ is the only partition of $n = 0$.
\end{definition}
Partitions can be represented pictorially with Young diagrams.
\begin{definition}[Young diagrams]
    Let $\lambda \vdash n$. A \emph{Young diagram} of \emph{shape} $\lambda$ is a diagram consisting of $n$ boxes, arranged into $\ell(\lambda)$ left-justified rows, such that the $i$'th row contains $\lambda_i$ boxes. The box in the $i$'th row and $j$'th column is denoted by $(i, j)$.
\end{definition}
\begin{figure}[htbp]
    \centering
    \ytableausetup{boxsize=1.3em}

    \begin{minipage}{0.45\textwidth}
        \centering
        \ydiagram{6,3,1}
        \par\medskip
        $\lambda = (6,3,1)$
    \end{minipage}

    \caption{The Young diagram associated with a partition $\lambda = (6, 3, 1)$ of $n = 10$.}
    \label{fig:young-diagram}
\end{figure}

\begin{definition}[Standard Young tableau]
    A \emph{standard Young tableau} (SYT) of shape $\lambda \vdash n$ is a bijective labeling of the boxes of $\lambda$ with the numbers in $[n]$, such that the labels are strictly increasing rightwards along rows, and downwards along columns. The set of all SYTs of shape $\lambda$ is denoted $SYT(\lambda)$.
\end{definition}

\begin{definition}[Semistandard Young tableaux]
    Fix a positive integer $d$. A \emph{semistandard Young tableau} (SSYT) of shape $\lambda$ and alphabet $[d]$ is a labelling of the boxes of $\lambda$, such that the labels are weakly increasing rightwards along rows, and strictly increasing downwards along columns. The set of all SSYTs of shape $\lambda$ and alphabet $[d]$ is denoted $SSYT(\lambda, d)$.
    
    Enumerating how many times each number in $[d]$ appears in tableau $T$ gives a sequence known as the \emph{weight} of $T$, denoted as $w(T) = (w_1(T), \ldots, w_d(T))$. The label of box $(i, j)$ in a tableau $T$ is denoted by $T(i, j)$.
\end{definition}

\begin{definition}[Rational tableaux]
    Fix a positive integer $d$. For partitions $\lambda, \nu$ where $\ell(\lambda)+\ell(\nu) \leq d$, a \emph{rational tableau} of shape $(\lambda, \nu)$ and alphabet $[d]$ is a pair
        \[ (A, B) \in SSYT(\lambda, d) \times SSYT(\nu, d) \]
    such that for all $1 \leq i \leq d$, $|\{j: A(j, 1) \leq i\}|+|\{j: B(j, 1) \leq i\}| \leq i$. We denote the set of all rational tableaux of shape $(\lambda, \nu)$ and alphabet $[d]$ by $RT(\lambda, \nu, d)$.
\end{definition}

\begin{figure}[htbp]
    \centering
    \ytableausetup{boxsize=1.3em}

    \begin{minipage}{0.30\textwidth}
        \centering
        \begin{ytableau}
            1 & 2 & 4 & 8 \\
            3 & 5 & 6 \\
            7
        \end{ytableau}
    \end{minipage}
    \hfill
    \begin{minipage}{0.30\textwidth}
        \centering
        \begin{ytableau}
            1 & 1 & 2 & 3 & 4 & 4\\
            2 & 2 \\
            4
        \end{ytableau}
    \end{minipage}
    \hfill
    \begin{minipage}{0.35\textwidth}
        \centering
        \(\left(\;
        \vcenter{\hbox{\begin{ytableau}
            1 & 1 & 1 & 3 & 6 \\
            3 & 4 \\
            4
        \end{ytableau}}}
        \, ,\;
        \vcenter{\hbox{\begin{ytableau}
            2 & 3 & 4 & 4 \\
            5 & 5 \\
            6 & 6
        \end{ytableau}}}
        \;\right)\)
    \end{minipage}

    \caption{On the left, a standard Young tableau of shape $(4, 3, 1)$. In the middle, a semistandard Young tableau of shape $(6, 2, 1)$ with alphabet $[4]$ and weight $(2, 3, 1, 3)$. On the right, a rational tableau of shape $((5, 2, 1), (4, 2, 2))$ with alphabet $[6]$.}
    \label{fig:young-tableaux}
\end{figure}

\begin{definition}[Polynomial and rational representations]
    Let $(\mu, V)$ be a representation of $U(N)$. Then $\mu$ is a \emph{polynomial representation} if we can pick a basis for $V$ such that the entries of the matrix $\mu(U)$ are polynomials in the entries of $U$ for all $U \in U(N)$. Similarly, $\mu$ is a \emph{rational representation} if we can pick a basis for $V$ such that the entries of the matrix $\mu(U)$ are polynomials in the entries of $U$ and $\overline{U}$ for all $U \in U(N)$.
\end{definition}
Note that a representation being polynomial or rational is a basis-independent notion.

\begin{theorem}
    The polynomial irreps of $U(d)$ are in bijection with partitions $\lambda$ of length at most $d$. The irrep corresponding to $\lambda$ is denoted by $(\rho_\lambda, V^d_{\lambda})$, where $V^d_{\lambda}$ is called the \emph{Schur module}.
    
    Similarly, the rational irreps of $U(d)$ are in bijection with pairs of partitions $(\lambda, \nu)$ satisfying $\ell(\lambda)+\ell(\nu) \leq d$. The irrep corresponding to $(\lambda, \nu)$ is denoted by $(\rho_{\lambda, \nu}, V^d_{\lambda, \nu})$.
\end{theorem}

\begin{theorem}
    There is a basis of $V^d_{\lambda}$ for which each basis element is bijectively associated with an SSYT of shape $\lambda$ and alphabet $[d]$, so that $\dim(V^d_{\lambda}) = |SSYT(\lambda, d)|$.
\end{theorem}

\begin{theorem}[{\cite[Proposition 2.4]{Ste87}}]
    There is a basis of $V^d_{\lambda, \nu}$ for which each basis element is bijectively associated with a rational tableau of shape $(\lambda, \nu)$ and alphabet $[d]$, so that $\dim(V^d_{\lambda, \nu}) = |RT(\lambda, \nu, d)|$.
\end{theorem}

\begin{definition}[Rational Schur polynomial]
    Let $x_1, \ldots, x_d$ be indeterminates. For partitions $\lambda$ and $\nu$ where $\ell(\lambda) + \ell(\nu) \leq d$, the \emph{rational Schur polynomial} $s_{\lambda, \nu}(x_1, \ldots, x_d)$ is the symmetric Laurent polynomial defined as 
        \[ s_{\lambda, \nu}(x_1, \ldots, x_d) := \sum_{(T, T') \in RT(\lambda, \nu, d)} \prod_{i = 1}^d x_i^{w_i(T)-w_i(T')}. \]
\end{definition}

\begin{theorem}
Let $M\in U(d)$ have eigenvalues $\alpha_1, \ldots, \alpha_d$. Then $\chi_{\rho_{\lambda, \nu}}(M) = s_{\lambda, \nu}(\alpha_1, \ldots, \alpha_d)$.
\end{theorem}

\begin{theorem}[Mixed Schur-Weyl duality, \cite{Ngu23,GBO23}]
Fix $2r < N$ and consider the Hilbert space $V = (\bbC^N)^{\otimes r} \otimes (\overline{\bbC^N})^{\otimes r}$. Define the representation $(\phi^N_r, V)$ of $U(N)$ where
    \[ \phi^N_r(U) \ket{x_1, \ldots, x_r, x_{r+1}, \ldots, x_{2r}} = (U^{\otimes r} \otimes \overline{U}^{\otimes r}) \ket{x_1, \ldots, x_r, x_{r+1}, \ldots, x_{2r}}. \]
Then $\phi_r^N$ has the following decomposition:
    \[ \phi_r^N \cong \bigoplus_{k = 0}^r \bigoplus_{\lambda, \nu \vdash k} \rho_{\lambda, \nu} \otimes I_{m_{\lambda,\nu}^{(r, N)}}. \]
\end{theorem}

\section{Random-Oracle Unitary Synthesis}\label{sec:rom_synthesis}
We write $N=2^n$ for the dimension of the target system, $M=2^m$ for the total Hilbert-space dimension, and $M_0=2^{m_0}$ for the number of oracle inputs.
We begin by recalling the definition of unitary synthesis, following the convention from \cite{LMW24}.

\begin{definition}[Unitary Synthesis]\label{def:unitary_synthesis}
An $n$-qubit unitary synthesis algorithm with error $\eps$ which uses $k$ queries and $m$ qubits is a list of unitaries $V_0, \dots, V_k \in U(M)$ such that for any unitary $U \in U(N)$, there exists $m_0 \leq m$ and a function $f:\{0, 1\}^{m_0} \to \{0, 1\}$ such that the unitary
\begin{equation*}
W = V_k \mathcal{O}_f V_{k-1} \dots V_1 \mathcal{O}_f V_0,
\end{equation*}
satisfies
\begin{equation*}
D_{\diamond}(\pi \circ \Phi_W \circ i, \Phi_U) \leq \eps(n).
\end{equation*}
Here, $\pi\circ\Phi_W\circ i$ appends $m-n$ zero
ancillas to the input state, applies $W$, and discards
the last $m-n$ qubits, as defined in the notation section.
\end{definition}

\begin{remark}
A unitary synthesis algorithm is \emph{efficient}
if $k,m=\poly(n)$.
We call its approximation guarantee \emph{nontrivial}
if $1-\eps\ge 1/\poly(n)$.
\end{remark}

We now state the definition of a ROM-PRU, given in \cite{BY26}.
\begin{definition}[ROM-PRU]\label{def:rom-pru}
A $T$-query-secure $n$-qubit pseudorandom unitary in the random-oracle model, implemented in $k$ queries with oracle input length $m_0$ and $m$ total qubits, with design error $\delta$ and implementation error $\epsilon$, is a sequence of unitaries $V_0, \dots, V_k \in U(M)$ such that for a random function $f:\{0, 1\}^{m_0} \to \{0, 1\}$, the random variable
\begin{equation*}
W = V_k \mathcal{O}_f V_{k-1} \dots V_1 \mathcal{O}_f V_0
\end{equation*}
satisfies the following condition: there exists a joint random variable $(\Phi_1, \Phi_2)$ such that:

\begin{enumerate}
\item The marginal distribution of $\Phi_1$ is equal to that of $\pi \circ \Phi_W \circ i$.
\item The marginal distribution of $\Phi_2$ is that of $\Phi_U$ for an ensemble such that no algorithm making at most $T$ adaptive forward queries to $U$ can distinguish it from Haar with advantage greater than $\delta$.
\item $D_\diamond(\Phi_1, \Phi_2) \leq \eps$.
\end{enumerate}
We additionally require that each oracle $f$ determine a single target unitary $U_f\in U(N)$ in the coupling. More precisely, there must exist a deterministic assignment $f\mapsto U_f$ such that $(\Phi_1,\Phi_2)$ has the joint distribution of $(\pi\circ\Phi_W\circ i,\Phi_{U_f})$ for uniformly random $f$, where $W$ is the circuit above evaluated at $f$.

Conditions (1)--(3) alone are equivalent to requiring that the Wasserstein (diamond) $\infty$-distance between the distribution of $\pi \circ \Phi_W \circ i$ and that of $\Phi_U$ be at most $\eps$.
\end{definition}

\begin{remark}
The requirement that each oracle $f$ determine a single target unitary $U_f$ follows the definition of Brakerski and Yuen~\cite{BY26}. Our lower bounds do not require this restriction: they hold slightly more generally for arbitrary couplings satisfying conditions (1)--(3), even when the target unitary is not determined by $f$.

The ROM-PRU model encapsulates one strategy to build pseudorandom unitaries.
Say that $V_0, \dots, V_k$ are efficient unitaries, and that $\eps^{-1} = \Omega(\text{exp}(2\kappa))$ and $\delta^{-1}, T = \Omega(\text{exp}(\kappa))$ for some $\kappa$.
Then, one can use a ROM-PRU to build a pseudorandom unitary secure to $\text{exp}(\kappa)$ queries---by replacing the random function with a quantum-secure pseudorandom function, such that the distinguishing advantage with $\text{exp}(\kappa)$ queries is $O(\text{exp}(-\kappa))$.

While there have been many recent PRU constructions \cite{MH25,SHH25}, they all lack a property called \emph{scalability}. In particular, these constructions are only able to achieve security parameter $\kappa \leq n/2$ with $k = \poly(n)$ queries. This means that to achieve a desired security parameter one must necessarily increase the number of qubits that the PRU applies to. Of course, it could be desirable to obtain security that scales \emph{independently} of the number of qubits, especially at low dimensions.
A \emph{scalable PRU} is accordingly defined to be a pseudorandom unitary that can achieve $\kappa$ bits of security with $k = \poly(n, \kappa)$ queries.
\end{remark}

A priori, it is not clear if there is a formal relationship between unitary synthesis and the existence of ROM-PRUs. Intuitively, both describe the type of unitaries that can be implemented with circuits of the form $V_k \mathcal{O}_f \dots \mathcal{O}_f V_0$.  However, while unitary synthesis requires a worst-case condition over all unitaries $U \in U(N)$ but allows the corresponding function $f$ to depend on $U$, in a ROM-PRU, the random variable $V_k \mathcal{O}_f \dots \mathcal{O}_f V_0$ for a random function $f$ must be close to a $T$-design.

Despite these differences, Brakerski and Yuen recently showed that these two notions can be connected in a rigorous manner:
\begin{theorem}[{\cite[proof of Lemma 4.7 and Remark 4.9]{BY26}}]
There are universal constants $c_1,c_2,c_3>0$
such that the following holds.
Suppose there exists an $n$-qubit ROM-PRU secure
against $T\ge c_3 N^2$ queries, implemented using
$m$ total qubits, oracle input length $m_0$,
and $k$ oracle queries, with design error
$\delta\le c_1$ and implementation error
$\eps\le c_2$.
Then there exists a unitary synthesis algorithm using
at most $2m-n+1$ total qubits, oracle input length
$m_0+1$, and $2k$ oracle queries, with implementation
error at most $1/2$.
\end{theorem}
This theorem shows in particular that scalable PRUs in the ROM-PRU model would imply unitary synthesis---the authors use this result as evidence that such a construction would be highly nontrivial.
But since there are no known lower bounds to unitary synthesis algorithms that make more than one query, this result does not unconditionally rule out the possibility of scalable PRUs in the ROM-PRU model.

Inspired by the connection, we will define the notion of random-oracle unitary synthesis below, as an average-case version of unitary synthesis.
Just like the ROM-PRU model, average-case unitary synthesis only queries a uniformly random function, instead of an arbitrary one---as we see, this will allow us to apply similar techniques to analyze both settings, and eventually deduce lower-bounds that we are unable to extend to the general worst-case regime.
In this section, we define a natural version of unitary synthesis when given access to a random-oracle.
This question is closely related to the ROM-PRU model of \cite{BY26}; indeed, by using our techniques for ruling out scalable ROM-PRUs, we can also demonstrate that efficient random-oracle unitary synthesis is impossible.

\begin{definition}[Random-Oracle Unitary Synthesis]\label{def:rom_synthesis}
An $n$-qubit \emph{random-oracle} unitary synthesis algorithm with error $\eps$, using $k$ queries, $m$ total qubits, and oracle input length $m_0\leq m$, is a list of unitaries $V_0, \dots, V_k \in U(M)$ with the following property: for the random variable
\begin{equation*}
W = V_k \mathcal{O}_f V_{k-1} \dots V_1 \mathcal{O}_f V_0,
\end{equation*}
where $f: \{0, 1\}^{m_0} \rightarrow \{0, 1\}$ is a uniformly random function on the first $m_0$ qubits, there exists a joint random variable $(\Phi_1, \Phi_2)$ such that:

\begin{enumerate}
\item The marginal distribution of $\Phi_1$ is equal to that of $\pi \circ \Phi_W \circ i$.
\item The marginal distribution of $\Phi_2$ is equal to the distribution of $\Phi_U$ for a Haar random unitary $U \in U(N)$.
\item $\bbE[D_\diamond(\Phi_1, \Phi_2)] \leq \eps$.
\end{enumerate}

Equivalently to conditions (1)-(3), the Wasserstein (diamond) 1-distance of $\pi \circ \Phi_W \circ i$ and $\Phi_U$ is at most $\eps$.
\end{definition}

This random-oracle definition of unitary synthesis is very similar to the standard one (Definition \ref{def:unitary_synthesis}), except rather than measuring the diamond distance between the channels induced by $W$ and $U$, one considers the analogous Wasserstein distance between their random variables.

Likewise, it is similar to the definition of a ROM-PRU, except (1) the approximate $T$-design is replaced with a Haar random distribution, (2) the constraint $\bbE[D_\diamond (\Phi_1, \Phi_2)] \leq \eps$ replaces $D_\diamond (\Phi_1, \Phi_2) \leq \eps$, and (3) the coupling need not be induced by a deterministic assignment $f\mapsto U_f$.

Note that this makes a random-oracle unitary synthesis algorithm and a ROM-PRU technically incomparable: while a Haar random unitary is a perfect $T$-design for any $T$, bounding the expected diamond distance (equivalently, considering the Wasserstein 1-distance instead of the $\infty$-distance) is a weaker criterion.
In particular, a random-oracle unitary synthesis algorithm can implement just a $\frac{1}{\poly(n)}$ fraction of unitaries to non-negligible precision.

\section{Impossibility of Random-Oracle Unitary Synthesis} \label{sec:lower_bound}
Our main result in this section is the impossibility of random-oracle unitary synthesis.

\begin{theorem}[Random-Oracle Unitary Synthesis Query Bound]\label{thm:main_1}
Suppose that there is an $n$-qubit random-oracle
unitary synthesis algorithm with error $\eps$
satisfying $1-\eps=\omega(N^{-2})$, which makes
$k$ queries and uses $m$ total qubits.
Then
\[
k^2(1-\eps)^{-6}=\Omega(N).
\]
In particular, if $1-\eps\ge 1/\poly(n)$, then
\[
k\ge \frac{\sqrt{N}}{\poly(n)}.
\]
\end{theorem}

The bound shows that the regime $k = \poly(n)$ and error $\eps = 1 - \frac{1}{\poly(n)}$ is impossible, for any choice of $m$.
Notably, this $m$-independent result cannot be true in the original model of unitary synthesis---there is a known algorithm with $k = 1$, $m = \widetilde{O}(N^2)$, and constant $\eps < 1$ \cite{BV93,LMW24}.
Our result therefore demonstrates that the presence of structure in the oracle queries is necessary for efficient unitary synthesis, even in the regime of unbounded ancilla qubits. 

Using similar techniques, we demonstrate that no query-efficient ROM-PRU can obtain security beyond $\widetilde{O}(N)$ queries under the error assumptions below. 

\begin{theorem}[ROM-PRU Query Bound]\label{thm:main_2}
Fix $\gamma>0$. Suppose there is an $n$-qubit ROM-PRU
secure against $T=\Omega(N^{1+\gamma})$ queries,
implemented using $k\ge1$ oracle queries and $m$ total
qubits, with implementation error $\eps$ and design
error $\delta$. Assume $0\le\eps,\delta<1$ and
$1-\eps=\omega(N^{-2})$. Then
\[
k^2(1-\delta)^{-2}(1-\eps)^{-6}
=\Omega\left(N^{\min\{\gamma,1\}}\right).
\]
In particular, if $(1-\eps)^{-1}$ and $(1-\delta)^{-1}$
are polynomially bounded in $n$, then
\[
k\ge
\frac{N^{\min\{\gamma,1\}/2}}{\poly(n)}.
\]
\end{theorem}
\begin{remark}[Query-access models]
The lower bound holds even for standard non-adaptive designs (Definition~\ref{def:standard_nonadaptive_design}), whereas the upper bound in Theorem~\ref{thm:linear_design} uses strong adaptive designs (Definition~\ref{def:strong_adaptive_design}).
\end{remark}
\begin{remark}
The proof applies immediately to a slightly strengthened statement, showing that $k^2(1-\delta)^{-2}(1-\epsilon)^{-6}= \Omega\left(\min\left(N, \frac{T}{N}\right)\right)$. 
This version precludes efficient, nontrivial ROM-PRUs from achieving $T$-designs whenever $\frac{T}{N}$ is superpolynomial in $n$. 
\end{remark}

\subsection{Clean Synthesis}
The intermediate step to our proof is to reformulate the central question.
We write $M=2^m$ for the total Hilbert-space dimension and $M_0=2^{m_0}$ for the number of oracle inputs. 
\begin{theorem}\label{thm:clean_bound}
Let $V_0, \dots, V_k \in U(M)$ be a sequence of unitaries, and let $J$ be the isometry taking an $n$-qubit state $\ket{\psi}$ to $\ket{\psi}\ket{0^{m-n}}$. 
For a uniformly random function $f:\{0, 1\}^{m_0} \to \{0, 1\}$, define the random variable 
\begin{equation*}
A_f = J^\dag V_k \mathcal{O}_f \dots \mathcal{O}_f V_0 J. 
\end{equation*}
For every $0<r\le1$, there exists a randomized algorithm
$\mathsf{Resample}$ that takes as input $f$ and produces
a subset $S\subseteq\{0,1\}^{m_0}$ and a function
$f_S:\{0,1\}^{m_0}\to\{0,1\}$ such that 
\begin{equation*}
H( f_S | f ) \geq M_0 -  \frac{32k^2N}{r^2}
\end{equation*}
and
\begin{equation*}
\bbE_{f, (S, f_S) \gets \mathsf{Resample}(f)}[\lVert A_f - A_{f_S} \rVert^2_{F}] \leq 2Nr^2.
\end{equation*}

\end{theorem}
The $N \times N$ submatrix $A_f$ of the circuit $V_k \mathcal{O}_f\dots \mathcal{O}_fV_0$ can be understood as approximately implementing a unitary $U$. 
Indeed, say that $A_f$ is near some unitary $U$ on $n$ qubits---then, $V_k \mathcal{O}_f\dots \mathcal{O}_fV_0$ approximately takes a state $\ket{\psi}\otimes \ket{ 0^{m-n}}$ to $U\ket{\psi} \otimes \ket{0^{m-n}}$.
This is an approximate \emph{clean} implementation of $U$, since the ancilla state is unchanged by the unitary. 

Theorem \ref{thm:clean_bound} is a compelling argument that \emph{cleanly} synthesizing many unitaries with access to a random oracle $f$ cannot be done efficiently. 
The theorem states that one can stochastically alter $f$ at almost all inputs, while the value of $A_f$ changes by a small amount. 
Indeed, the conditional entropy of $f_S$ with respect to $f$ is nearly $M_0$, meaning it is almost fully resampled. 
Informally, $A_f$ can only contain $O\left(\frac{N k^2}{r^2}\right)$ bits of entropy, since the other degrees of freedom only perturb $A_f$ by a bounded Frobenius distance. 

This result nearly yields our desired lower bounds.
Indeed, it is known that a ROM-PRU secure against $T=\Omega(N^{1+\gamma})$ queries for $0<\gamma\le1$ must have a support with entropy of $\Omega\left(N^{1+\gamma} \right)$, while a random-oracle unitary synthesis algorithm requires an entropy of $\Omega(N^2)$.
However, there are two limitations to this approach. 
Firstly, the simplest entropy bounds are not necessarily robust to perturbing the support by bounded Frobenius distance---it is hypothetically possible that doing so could collapse the entropy. 
This issue is eventually handled, as we show how to carry forward the argument in a manner that is robust to perturbations. 

Secondly, and more troublingly, a ROM-PRU and a random-oracle unitary synthesis algorithm may implement unitaries \emph{dirtily}, changing the ancilla state from $\ket{0^{m-n}}$ in the process of the implementation.
The existence of such an algorithm cannot be immediately invalidated by a lower bound against clean synthesis. 
To resolve this problem, we demonstrate that a dirty synthesis algorithm can still indirectly imply a contradiction to Theorem \ref{thm:clean_bound}.
In particular, we will show that the action of a dirty synthesis channel $\Phi$ on a Pauli observable $P \otimes I$ produces an $N \times N$ output that contradicts Theorem \ref{thm:clean_bound}.
Hence, while our intermediate theorem more straightforwardly proves the nonexistence of a clean synthesis algorithm, the additional technical trick allows us to complete our main theorems in full generality. 

\begin{definition}\label{def:weight}
Consider any sequence of unitaries $V_0, \dots, V_k \in U(M)$, function $f: \{0, 1\}^{m_0}\to \{0, 1\}$ and bitstring $x \in \{0, 1\}^{m_0}$. 
Let $\Pi_x := \ketbra{x}{x} \otimes I_{m-m_0}$ be the projector onto the first $m_0$ qubits being $\ket{x}$. Define the influence $w_f(x)$ to be 
\begin{equation*}
w_f(x) = \sum_{i = 1}^k (\lVert L_{f, i} \Pi_x \rVert^2_F+\lVert \Pi_x R_{f, i} \rVert^2_F),
\end{equation*}
where $i=1,\dots,k$ indexes the oracle calls from right to left,
\begin{equation*}
L_{f, i} = J^\dag V_k \mathcal{O}_f \dots V_{i} 
\end{equation*}
and 
\begin{equation*}
R_{f, i} = V_{i-1} \mathcal{O}_f \dots V_0 J .
\end{equation*}
\end{definition}
This quantity is a measure of the extent to which the output of $f$ at any given input $x \in \{0, 1\}^{m_0}$ influences the final block $A_f$. 
We can record a simple lemma which bounds the sum of the influences over all $x$ for a given function $f$:
\begin{lemma}\label{lem:weight_sum_bound}
For any function $f$ as in Definition \ref{def:weight}, 
\begin{equation*}
\sum_{x \in \{0, 1\}^{m_0}} w_f(x) = 2kN. 
\end{equation*}
\end{lemma}
\begin{proof}
We can expand the sum as
\begin{align*}
    \sum_x w_f(x) &= \sum_{x \in \{0, 1\}^{m_0}}\sum_{i = 1}^k (\lVert L_{f, i} \Pi_x \rVert^2_F+\lVert \Pi_x R_{f, i} \rVert^2_F)\\
    &=\sum_{i = 1}^k \left(\sum_{x \in \{0, 1\}^{m_0}}\lVert L_{f, i} \Pi_x \rVert^2_F+\sum_{x \in \{0, 1\}^{m_0}}\lVert \Pi_x R_{f, i} \rVert^2_F\right)\\
    &= \sum_{i = 1}^k (\lVert L_{f, i}\rVert^2_F+\lVert R_{f, i} \rVert^2_F) \\
    &= \sum_{i = 1}^k \left(\Tr L_{f, i}L_{f, i}^{\dagger} + \Tr R_{f, i} R_{f, i}^{\dagger}\right) \\
    &= \sum_{i = 1}^k (N+N) = 2kN.
\end{align*}
In the third equality, we use the fact that the operators $L_{f, i}\Pi_x$ (resp. $\Pi_x R_{f, i}$) for different values of $x$ are orthogonal under the Hilbert-Schmidt inner product.
\end{proof}
 
\begin{lemma}\label{lem:single_update_bound}
Consider any sequence of unitaries $V_0, \dots, V_k \in U(M)$, function $f: \{0, 1\}^{m_0} \to \{0, 1\}$ and bitstring $x \in \{0, 1\}^{m_0}$. 
Let $f'$ be the same function, but with $f'(x) = 1-f(x)$.
Then $\lVert A_{f'} -A_f\rVert_F^2 \leq 4w_f(x) w_{f'}(x)$. 
\end{lemma}
\begin{proof}
Assume without loss of generality that $f(x) = 1$.
We can first expand $A_{f'} - A_f$ as a telescoping sum, by defining a hybrid $H_i$, for $0 \leq i \leq k$, in which the first $i$ oracle calls (counted from right to left) use $f'$ and the remaining queries use $f$. Thus $H_0=A_f$ and $H_k=A_{f'}$, giving the identity
\begin{equation*}
A_{f'} - A_f  = \sum_{i = 1}^{k} (H_i - H_{i-1}).
\end{equation*}
Since $H_i = L_{f, i}\mathcal{O}_{f'} R_{f', i}$ and $H_{i-1}=L_{f, i}\mathcal{O}_{f} R_{f', i}$ for $1 \leq i \leq k$, the sum expands to
\begin{align*}
A_{f'} - A_f &=\sum_{i = 1}^{k}  \left(L_{f, i}\mathcal{O}_{f'} R_{f', i} - L_{f, i}\mathcal{O}_f R_{f', i}\right)\\
&= 2\sum_{i = 1}^{k} L_{f, i} \Pi_x R_{f', i}.
\end{align*}
Now, we may rewrite this sum as 
\begin{equation*}
2\overbrace{\begin{bmatrix}
L_{f, 1} \Pi_x & \dots & L_{f, k} \Pi_x
\end{bmatrix}}^{\mathbf{L}}
\overbrace{\begin{bmatrix}
\Pi_xR_{f', 1}   \\
\dots \\
\Pi_x R_{f', k} 
\end{bmatrix}}^\mathbf{R'}.
\end{equation*}
Notice that $\lVert \mathbf{L} \rVert_F^2 = \sum_{i=1}^k \lVert L_{f, i} \Pi_x\rVert^2_F \leq w_f(x)$, and likewise $\lVert \mathbf{R'} \rVert_F^2 \leq w_{f'}(x)$.
In particular,
\begin{align*}
   4\left\lVert \sum_{i = 1}^{k} L_{f, i} \Pi_x R_{f', i}\right\rVert^2_F &= 4\lVert \mathbf{LR'} \rVert^2_F \\
   &\leq 4\lVert \mathbf{L} \rVert_F^2 \lVert \mathbf{R'} \rVert_F^2 \\
   &\leq 4w_f(x)w_{f'}(x),
\end{align*}
completing the proof.
\end{proof}
With this lemma, we are now prepared to give the main proof of this section. 
\begin{proof}[Proof of Theorem \ref{thm:clean_bound}]
We will define a random process that produces the random variable $S$ in the theorem statement. 
By analyzing this random process, we will be able to conclude both conditions,
\begin{equation*}
H( f_S | f ) \geq M_0 -  \frac{32k^2N}{r^2}
\end{equation*}
and 
\begin{equation*}
\bbE[\lVert A_f - A_{f_S} \rVert^2_{F}] \leq 2Nr^2.
\end{equation*}
Our random process takes a parameter $\tau = \frac{r^2}{8k}$ as input, and proceeds as follows:
\begin{enumerate}
\item Let $f_0 \coloneqq f$, and initialize $S = \emptyset$.
\item Iterate over the bitstrings in $\{0, 1\}^{m_0}$ in some fixed order $x_1, \dots, x_{M_0}$. For $i = 1, \dots, M_0$. 
\begin{enumerate}
\item Set $f_i \coloneqq f_{i-1}$. Let $f_{i-1}^-$ and $f_{i-1}^+$ denote $f_{i-1}$, with the value of $x_i$ set to $0$ and $1$ respectively. 
If $w_{f_{i-1}^-}(x_i) \leq \tau$ and $w_{f_{i-1}^+}(x_i) \leq \tau$, resample the output of $f_i$ on $x_i$ randomly, and add $x_i$ to $S$. 

\end{enumerate}
\item Set $f_S \coloneqq f_{M_0}$.
\end{enumerate}

The procedure greedily resamples the values of the original function $f$, when it observes that the index of interest has a sufficiently small influence over the output.
Notice that each $f_i$ has the same distribution as $f$, that of a uniformly random function. 
Indeed, $f_0 = f$ is just a uniformly random function. Conditioned on $f_0^+$ and $f_0^-$, the value of $f_0$ at $x_1$ is uniformly random. 
Therefore, in the update step to $f_1$, the value either remains uniformly random, or is resampled. Neither of these choices changes the distribution, and this argument can be applied inductively to any $f_i$. 
So, the expectation of random variable $M_0-|S|$ can be expressed as
\begin{align*}
\bbE[M_0-|S|] &= \sum_{i=1}^{M_0} \Pr[w_{f_{i-1}^+}(x_i) > \tau \text{ or }w_{f_{i-1}^-}(x_i) > \tau]\\
&\leq 2\sum_{x \in \{0, 1\}^{m_0}}  \Pr[w_f(x) > \tau],
\end{align*}
where the factor of two arises by a union bound over $f_{i-1}^+$ and $f_{i-1}^-$, or equivalently $f_{i-1}$ and $f_{i-1}$ with its output flipped on $x_i$---which are both uniformly random. 
This can be further bounded as
\begin{align}
2\sum_{x \in \{0, 1\}^{m_0}}  \Pr[w_f(x) > \tau]&\leq \frac{2}{\tau}\sum_{x \in \{0, 1\}^{m_0}}  \bbE[w_f(x)]\\
&= \frac{2}{\tau}\bbE\left[\sum_x w_f(x)\right]\\
&= \frac{4kN}{\tau}\\
&=\frac{32k^2N}{r^2}.\label{eq:weight_lower_bound}
\end{align}
The first inequality is an application of Markov's inequality, while the second applies Lemma \ref{lem:weight_sum_bound}.

Since each input is visited exactly once, the initial
function $f$ and the final function $f_S$ determine
every intermediate function: $f_i(x_j)=f_S(x_j)$ for
$j\le i$, while $f_i(x_j)=f(x_j)$ for $j>i$.
Thus the full trajectory has the same conditional
entropy as its final function. By the chain rule,
\begin{align*}
H(f_S\mid f)
&=H(f_1,\ldots,f_{M_0}\mid f)\\
&=\sum_{i=1}^{M_0}
H(f_i\mid f_{i-1},\ldots,f_0).
\end{align*}
Note first that
\begin{equation*}
H(f_i | f_{i-1}, \dots, f_0) = H(f_i | f_{i-1}).
\end{equation*}
Indeed, the sequence $(0, f_0), \dots, (M_0, f_{M_0})$ forms a Markov chain, so the dependence of $f_i$ on previous functions simplifies to a dependence on $f_{i-1}$. 
We can now consider whether $x_i \in S$ and split into two cases. 
First, $H(f_i | f_{i-1}, x_i \in S) = 1$, since in this case $f_i$ has one new uniformly random bit. 
Meanwhile, $H(f_i | f_{i-1}, x_i \notin S) = 0$, since $f_i = f_{i-1}$. 
Hence $H(f_i | f_{i-1}) = \Pr[x_i \in S]$, and therefore
\begin{align*}
H(f_S | f) &= \sum_{i = 1}^{M_0} \Pr[x_i \in S]\\
&= \bbE[|S|]\\
&\geq M_0 - \frac{32k^2N}{r^2},
\end{align*}
where we use Equation \ref{eq:weight_lower_bound} to finish the desired inequality. 

We now state a lemma that gives a triangle inequality for the expected Frobenius distance squared, albeit with an extra factor of 2---this lemma will be crucial for the theorem. 
\begin{lemma}\label{lemma:triangle_inequality}
\begin{equation*}
\bbE[\lVert A_f - A_{f_S} \rVert^2_{F}] \leq 2\sum_{i = 1}^{M_0}\bbE\left[ \lVert A_{f_i} - A_{f_{i-1}}\rVert^2_F \right]
\end{equation*}
\end{lemma}
First, let us see how assuming the lemma yields the result. 
Notice that each term in 
\begin{equation}\label{eq:sum_bound}
2\sum_{i = 1}^{M_0}\bbE\left[ \lVert A_{f_i} - A_{f_{i-1}}\rVert^2_F \right],
\end{equation}
is either zero since $x_i \notin S$, or we have the assurance that $w_{f_{i-1}}(x_i) \leq \tau$. 
So, using Lemma \ref{lem:single_update_bound}, we obtain the bound
\begin{align*}
\sum_{i = 1}^{M_0}\bbE\left[ \lVert A_{f_i} - A_{f_{i-1}}\rVert^2_F \right]&\leq \sum_{i = 1}^{M_0}\bbE\left[\mathds{1}_{x_i \in S} \cdot 4w_{f_{i-1}}(x_i)w_{f_i}(x_i)\right]
\\&\leq 4\tau\sum_{i =1}^{M_0} \bbE[w_{f_i}(x_i)]  \\
&\leq 4\tau \sum_{i = 1}^{M_0} \bbE[w_{f}(x_i)]\\
&\leq 8\tau k N\\
&\leq Nr^2.
\end{align*}
The third inequality notes that since every $f_i$ has the same distribution, we can replace them in the expression by a single uniformly random $f$. 
We then apply Lemma \ref{lem:weight_sum_bound}, since we have a sum of $w_f(x)$ over all $x \in \{0, 1\}^{m_0}$, and substitute in the value of $\tau$.
This inequality, after applying Lemma \ref{lemma:triangle_inequality}, yields the result. 

\begin{proof}[Proof of Lemma \ref{lemma:triangle_inequality}]
We can expand $A_{f_S}-A_f$ as 
\begin{equation*}
A_{f_S}-A_f = \sum_{i:x_i \in S} A_{f_{i}} - A_{f_{i-1}}.
\end{equation*}
We now define 
\begin{equation}
\overline{A}_{f_i} = \frac{A_{f_{i-1}^+} + A_{f_{i-1}^-}}
{2}.\end{equation}
For $x_i \in S$, $\overline{A}_{f_i}$ is the average of the two possible values for $f_i$, after it is sampled from $f_{i-1}$. 
As a consequence, 
\begin{equation}
\mathbb{E}[A_{f_i} | f_{i-1},\dots, f_0,  x_i \in S] = \overline{A}_{f_i}. 
\end{equation}
Then we can rewrite
\begin{equation*}
A_{f_{i}} - A_{f_{i-1}} = \overbrace{(A_{f_i} - \overline{A}_{f_i})}^{T_i^{(1)}}  + \overbrace{(\overline{A}_{f_i} - A_{f_{i-1}})}^{T_i^{(2)}}. 
\end{equation*}
When $x_i \notin S$, we define $T_i^{(1)}, T_i^{(2)} = 0$. 
By their definitions, note that when $x_i \in S$, 
\begin{equation*}
\lVert T_i^{(1)}\rVert_F= \lVert T_i^{(2)}\rVert_F = \frac{1}{2}\big\lVert A_{f_{i-1}^+}-A_{f_{i-1}^-}\big\rVert_F,
\end{equation*}
since $A_{f_i} ,A_{f_{i-1}} \in  \{A_{f_{i-1}^+}, A_{f_{i-1}^-}\}$. 
From the conditional expectation of $\overline{A}_{f_i}$, it also follows that
\begin{equation*}
\mathbb{E}[T_i^{(1)} | f_{i-1}, \dots, f_0] = 0,
\end{equation*}
and therefore for any $i > j$, 
\begin{equation*}
\bbE[T_i^{(1)} (T_j^{(1)})^\dag] = 0,
\end{equation*}
since $T_j^{(1)}$ depends only on $f_1, \dots, f_j$, implying that $\bbE[T_i^{(1)} \, \vert \, (T_j^{(1)})^\dag] = 0$. 

The Pythagorean theorem implies that since the $T_i^{(1)}$ are orthogonal with respect to this inner product, if $T^{(1)} = \sum_{i = 1}^{M_0}T_i^{(1)}$,
\begin{align*}
\bbE\left[\lVert T^{(1)}\rVert_F^2\right] &= \sum_{i=1}^{M_0}\bbE\left[\lVert T_i^{(1)}\rVert_F^2\right]\\
&= \frac{1}{4}\sum_{i = 1}^{M_0}\bbE\left[\mathds{1}_{x_i \in S} \lVert A_{f_{i-1}^+} - A_{f_{i-1}^-}\rVert^2_F \right]\\
&= \frac{1}{2}\sum_{i = 1}^{M_0}\bbE\left[\mathds{1}_{x_i \in S} \lVert A_{f_i} - A_{f_{i-1}}\rVert^2_F \right]\\
&\leq \frac{1}{2}\sum_{i = 1}^{M_0}\bbE\left[ \lVert A_{f_i} - A_{f_{i-1}}\rVert^2_F \right].
\end{align*}
The third equality uses that $f_i, f_{i-1} = f_{i-1}^{\pm}, f_{i-1}^{\mp}$ with $\frac{1}{2}$ probability, and otherwise they are equal. 

However, it is not so obvious how to repeat the argument for the operators $T_i^{(2)}$. 
The crucial observation is that our random process is fully symmetric with respect to time: the random variables are completely consistent with first sampling $f_{M_0}$ at random, and then resampling the outputs at $x_{M_0}, \dots, x_1$ in reverse order with the same condition. 

This can be verified by examining any pair $f_i, f_{i-1}$. If one starts with $f_{i-1}$ and randomly resamples $x_i$ if the condition on $f_{i-1}^+$ and $f_{i-1}^-$ is satisfied, it yields the exact same joint distribution as starting with $f_{i}$ and resampling $x_i$ if the condition on $f_{i-1}^+$ and $f_{i-1}^-$ (which are both functions of $f_i$) is satisfied. 
Furthermore, $f_{i+1}$ and $f_{i-1}$ are independent conditioned on $f_{i}$, since as previously mentioned $(i, f_i)$ forms a Markov chain. 
Therefore, since all the consecutive joint distributions are equal for the forward and backwards processes, the complete joint distributions are equal. 

We may therefore adopt a reversed perspective on the sampling process.
It follows that
\begin{equation*}
\mathbb{E}[A_{f_{i-1}} | f_{i}, \dots, f_{M_0},  x_i \in S] = \overline{A}_{f_i}.
\end{equation*}
Using a symmetrical argument, we deduce that 
\begin{equation*}
\mathbb{E}[T_i^{(2)} | f_{i}, \dots, f_{M_0}] = 0,
\end{equation*}
and therefore for any $i < j$,
\begin{align*}
\bbE[T_i^{(2)}(T_j^{(2)})^\dag] = 0.
\end{align*}
Writing $T^{(2)}=\sum_{i=1}^{M_0}T_i^{(2)}$, we therefore obtain the same bound
\begin{equation*}
\bbE\left[\lVert T^{(2)}\right\rVert_F^2] \leq \frac{1}{2}\sum_{i = 1}^{M_0}\bbE\left[ \lVert A_{f_i} - A_{f_{i-1}}\rVert^2_F \right].
\end{equation*}
We can now conclude
\begin{align*}
\bbE\left[\lVert A_{f_S} - A_f \rVert_F^2\right] &\leq 2\bbE\left[\lVert T^{(1)}\rVert ^2_F\right] + 2\bbE\left[\lVert T^{(2)}\rVert ^2_F\right]\\
&\leq 2\sum_{i = 1}^{M_0}\bbE\left[ \lVert A_{f_i} - A_{f_{i-1}}\rVert^2_F \right].
\end{align*}
The first inequality uses $A_{f_S}-A_f = T^{(1)} + T^{(2)}$ and the elementary bound $\lVert X+Y\rVert_F^2 \leq 2\lVert X\rVert_F^2+2\lVert Y\rVert_F^2$, which is the $L=2$ case of the bound used in the technical overview.
\end{proof}
Having proven the lemma, the theorem follows.
\end{proof}

To conclude this section, we use Theorem \ref{thm:clean_bound} to isolate a convenient lemma. 
\begin{corollary}\label{cor:net}
Let $m$ be the total number of qubits including ancilla, $m_0$ the oracle query length, and $n$ the number of qubits in the target subsystem. Let $V_0, \dots, V_k \in U(M)$ be a sequence of unitaries with $k\ge1$. 
Once again, we define the random variable
\begin{equation*}
A_f = J^\dag V_k \mathcal{O}_f \dots \mathcal{O}_f V_0 J,
\end{equation*}
where $f:\{0, 1\}^{m_0} \to\{0, 1\}$ is a uniformly random function. 
Then for any $0<d<2$ and $0<\zeta<1$, there exists a
finite set $\mathcal T$ of contractions such that
$\log_2 (|\mathcal T|) = O\left(\frac{N k^2}{\zeta^2 d^2}\right)$, and
\begin{equation*}
\Pr\left[\min_{C \in \mathcal{T}} \lVert C - A_f \rVert_F > d\sqrt{N} \right] \leq \zeta.
\end{equation*}
If every $A_f$ is Hermitian, then $\mathcal T$ can be
chosen to consist of Hermitian contractions.
\end{corollary}
\begin{proof}
Let $B_f = \{g: \lVert A_f - A_g \rVert_F \leq d\sqrt{N} \}$. 
We can write $q_f = \frac{|B_f|}{2^{M_0}}$. 
In other words, $q_f$ is the proportion of functions $g$ whose $A_g$ are within $d\sqrt{N}$ Frobenius distance of $A_f$. 

First, we apply Theorem \ref{thm:clean_bound} with $r = d\sqrt{\frac{\zeta}{6}}$. As a consequence, for a uniformly random $f$, $f_S$ satisfies 
\begin{align*}
\Pr\left[\lVert A_{f_S} - A_{f}\rVert_F > d\sqrt{N}\right]= \Pr\left[\lVert A_{f_S} - A_{f}\rVert_F^2 > d^2N\right] \leq \frac{2r^2}{d^2} = \frac{\zeta}{3}.
\end{align*}
Denote the event that $\lVert A_{f_S} - A_{f}\rVert_F \leq d\sqrt{N}$ by $E$. 
The above equation shows that $\Pr[\lnot E] \leq \frac{\zeta}{3}$. 
Furthermore, we have that $H( f_S | f ) \geq M_0 -  \frac{32k^2N}{r^2}$.
We now set $b = \frac{32k^2N}{r^2}$, and define $L = \frac{b+1}{\zeta/3}$. 
Let $Q$ be the event $q_f\leq2^{-L}$. 
The conditional entropy $H(f_S | f)$ decomposes as
\begin{align*}
H(f_S | f) &\leq H(f_S, E | f)\\
&= H(f_S | f, E) + H(E | f)\\
&\leq
\Pr[E\land Q]\,
H(f_S\mid f,\mathds{1}_E=1,\mathds{1}_Q=1)
+(1-\Pr[E\land Q])M_0+1\\
&\leq \Pr[Q \land E](M_0 - L) + (1-\Pr[E \land Q])M_0 + 1 \\
&\leq M_0 - \Pr[ Q \land E]L + 1. 
\end{align*}
In the second line, we use the chain rule of conditional entropy. 
The term $H(E | f)$ is bounded by 1 subsequently, since the event $E$ has at most one bit. 
In the third line, we split according to $E$ and $Q$.
Since $Q$ is determined by $f$, conditioning on $Q$
requires no additional entropy term. 
If $E \land Q$ occurs, then the conditional entropy is at most $M_0 - L$---indeed, in this case, $A_{f_S}$ is within
$d\sqrt{N}$ Frobenius distance of $A_f$, and there
are at most $2^{M_0-L}$ possible functions $f_S$
with this property. 
Otherwise, there is the trivial bound on the entropy of $f_S$ of $M_0$. 

Using the inequality $H( f_S | f ) \geq M_0 -  b$, we obtain
\begin{equation*}
M_0 - b \leq M_0 - \Pr[Q \land E]\frac{b+1}{\zeta/3} + 1,
\end{equation*}
which rearranges to yield
\begin{equation*}
\Pr[Q \land E] \leq \frac{\zeta}{3}. 
\end{equation*}

We may now select the set $\mathcal{T}$ to be a set of operators $A_{g_i}$ for $i = 1, \dots, K = \lceil 2^L\ln(\frac{3}{\zeta})\rceil$ independently sampled, uniformly random functions $g_i$. 
Since $|\mathcal T|\le K$ and
$L=3(b+1)/\zeta\ge3/\zeta$, we have
\[
\log_2|\mathcal T|
\le\log_2 K
=O\left(L+\log_2\ln(3/\zeta)+1\right)
=O(L)
=O\left(\frac{k^2N}{\zeta^2d^2}\right),
\]
where the last equality uses $k\ge1$ and $d<2$. 
Each of these operators is a block of a unitary,
and hence has operator norm at most one, as required.
If every $A_g$ is Hermitian, then so is every element
of $\mathcal T$.  
For each fixed $f$ and $i$, the definition of $q_f$ gives
\begin{equation*}
\Pr_{g_i}\left[
\lVert A_{g_i}-A_f\rVert_F > d\sqrt{N}
\right]=1-q_f.
\end{equation*}
If $Q$ does not occur, this probability is at most $1 - 2^{-L}$. Hence, 
\begin{align*}
\Pr\left[\min_{C \in \mathcal{T}} \lVert C - A_f \rVert_F > d\sqrt{N} \,\Big \vert\, \lnot Q\right] &\leq (1-2^{-L})^K\\
&\leq \left((1-2^{-L})^{-2^L}\right)^{\ln(\zeta/3)}\\
&\leq \frac{\zeta}{3}. 
\end{align*}
In the final inequality, we use that
$(1-x)^{-1/x}\geq e$ for $0<x<1$,
with $x=2^{-L}$. 
We conclude that 
\begin{align*}
\Pr\left[\min_{C \in \mathcal{T}} \lVert C - A_f \rVert_F > d\sqrt{N} \right] &\leq \frac{\zeta}{3} + \Pr[Q]\\
&\leq \frac{\zeta}{3} + \Pr[Q \land E] + \Pr[\lnot E]\\
&\leq \zeta.
\end{align*}
Here the probability is over both $f$ and the sampled
functions $g_1,\ldots,g_K$.
Averaging over the sampled functions shows that there
is a fixed choice of $g_1,\ldots,g_K$ for which the
failure probability over $f$ alone is at most $\zeta$.
Fix this choice to obtain the required set $\mathcal T$.
\end{proof}

\subsection{Extending to Dirty Synthesis}
We have now proven the central technical result, but it seemingly only applies to the \textit{clean} synthesis of unitaries. 
Indeed, say that a circuit $W_f = V_k \mathcal{O}_f \dots \mathcal{O}_f V_0$ takes $\ket{\psi} \otimes \ket{0^{m-n}}$ close to $U\ket{\psi} \otimes \ket{0^{m-n}}$. 
Then $U \approx J^\dag W_f J$, since the ancilla state does not change.
Hence, if a circuit could yield a clean implementation of a net of sufficiently many unitaries, then we could use it to contradict the net constructed by Corollary \ref{cor:net}.

On the other hand, if the channel $\pi \circ \Phi_{W_f} \circ i$ corresponding to $W_f$ has diamond distance $\epsilon \leq \frac{1}{2} - \frac{1}{\poly(n)}$ from $\Phi_U$ for a unitary $U \in U(N)$, we can consider the circuit
\begin{equation*}
W_f^\dag (P \otimes I) W_f,
\end{equation*}
for any fixed Pauli $P \in \mathcal{P}_n$. 
Since $W_f^\dag$ erases the garbage introduced by $W_f$, a triangle inequality yields that
\begin{equation}
W_f^\dag (P \otimes I) W_f \left(\ket{\psi} \otimes \ket{0^{m-n}}\right)
\end{equation}
has trace distance at most $1- \frac{1}{\poly(n)}$ from
\begin{equation}
(U^\dag P U\ket{\psi}) \otimes \ket{0^{m-n}},
\end{equation}
which is sufficient to apply the clean obstruction. 
However, a random-oracle unitary synthesis algorithm only requires an expected error of $1 - \frac{1}{\poly(n)}$, in which case the above argument no longer holds. 
However, the lemma below shows that $J^\dag W_f^\dag (P \otimes I) W_fJ$, in expectation over a $\textit{random}$ Pauli $P$, has trace overlap with $U^\dag P U$.
This is the crucial reduction that enables all our lower bounds. 

\begin{lemma}\label{lem:dirty_to_clean}
Consider some sequence of unitaries $V_0, \dots, V_k \in U(M)$. 
For some fixed boolean function $f: \{0, 1\}^{m_0} \to \{0, 1\}$, let
\begin{equation*}
W_f = V_k\mathcal{O}_fV_{k-1} \dots \mathcal{O}_f V_0.
\end{equation*}
Let $P$ be a uniformly random, $n$-qubit Pauli operator. Let $A_{f,P}$ be $J^\dag W_{f}^\dag (P \otimes I) W_f J$---i.e., the $N \times N$ block of $W_{f}^\dag (P \otimes I) W_f$ corresponding to all ancilla qubits equal to zero. 
Then we have that for any unitary $U \in U(N)$, 
\begin{equation*}
\bbE_{P}\left[\text{tr}\left(\frac{1}{N}A_{f, P} \cdot ( U ^\dag P U ) \right)\right] \geq 1 - D_\diamond (\pi \circ \Phi_{W_f} \circ i, \Phi_U).
\end{equation*}
\end{lemma}
This lemma takes a circuit $W_f$, and argues that \emph{if} $W_f$ is close to a unitary $U$ as a channel, then the action of $W_f$ on a random observable $P$ has significant trace overlap with the action of $U$ on $P$.
Notice that $W_f^\dag (P \otimes I) W_f $ is a $2k$-query circuit---we will eventually apply our clean synthesis obstruction to this circuit.
\begin{proof}
First, we consider two Choi states. The first is
\begin{equation*}
\ket{U} = \frac{1}{\sqrt{N}}(U \otimes I)\sum_{x \in \{0, 1\}^n} \ket{x} \otimes \ket{x},
\end{equation*}
which we can rewrite as
\begin{align*}
\frac{1}{\sqrt{N}}\sum_{x\in \{0, 1\}^n} U\ket{x}\otimes\ket{x}&=\sum_{x,y\in \{0, 1\}^n} U_{yx}\ket{y}\otimes\ket{x}\\
&=\sum_{y\in \{0, 1\}^n}\ket{y}\otimes U^T\ket{y}.
\end{align*}
The state $\frac{1}{\sqrt{N}}\sum_{x \in \{0, 1\}^n} \ket{x} \otimes \ket{x}$ has density matrix $\frac{1}{N^2}\sum_{P \in \mathcal{P}_n} P \otimes P^T$.
As a result, we may write
\begin{equation*}
\ket{U}\bra{U} = \frac{1}{N^2}\sum_{P \in \mathcal{P}_n} P \otimes (U^{\dag} PU)^T,
\end{equation*}
Now, we will repeat the same process for $W_f$.
We can consider the state
\begin{equation*}
\ket{W_f} = \frac{1}{\sqrt{N}}((W_f J) \otimes I)\sum_{x \in \{0, 1\}^n} \ket{x}\otimes \ket{x},
\end{equation*}
where $J$ is the isometry taking $\ket{x}$ to $\ket{x0^{m-n}}$. 
In this case, rather than a normal Choi state for $W_f$ on $2m$ qubits, we are preparing a state on $m+n$ qubits that only takes a superposition over bitstrings $x \in \{0, 1\}^n$ on $n$ qubits.
Using the same argument as for $U$, which applies to the isometry $W_fJ$ as well, we obtain that
\begin{equation*}
\ket{W_f} = \frac{1}{\sqrt{N}}(I \otimes (W_f J)^T)\sum_{y \in \{0, 1\}^m} \ket{y}\otimes \ket{y},
\end{equation*}
and therefore
\begin{equation*}
\ket{W_f}\bra{W_f} = \frac{1}{M N}\sum_{Q \in \mathcal{P}_m} Q \otimes (J^\dag W_f^\dag Q W_fJ)^T.
\end{equation*}
Let $\rho_{W_f}$ be equal to $\ket{W_f}\bra{W_f}$, with the last $m-n$ qubits of the first register traced out. 
In doing so, all Pauli operators $Q = P \otimes R$ where $R \neq I$ vanish, since $\text{tr}(R) = 0$ if $R \neq I$. 
Meanwhile, if $R=I$, $\text{tr}(R) = M/N$, so $\rho_{W_f}$ simplifies to
\begin{align*}
\rho_{W_f} &= \frac{1}{N^2}\sum_{P \in \mathcal{P}_n} P \otimes (J^\dag W_f^\dag (P \otimes I) W_f J)^T\\
&=  \frac{1}{N^2}\sum_{P \in \mathcal{P}_n} P \otimes A_{f, P}^T,
\end{align*}
Now using the orthogonality of the Pauli operators, along with the fact that $\text{tr}(P \cdot P) = N$, we obtain
\begin{align*}
\braket{U|\rho_{W_f}|U} &= \frac{1}{N^4}\text{tr}\left(\left(\sum_{P \in \mathcal{P}_n} P \otimes A_{f, P}^T\right)\left(\sum_{P \in \mathcal{P}_n} P \otimes (U^{\dag} PU)^T\right)\right)\\
&= \frac{1}{N^3} \sum_{P \in \mathcal{P}_n} \text{tr}\left(A_{f, P}(U^\dag P U)\right)\\
&= \bbE_{P}\left[\text{tr}\left(\frac{1}{N}A_{f, P}(U^\dag P U)\right)\right].
\end{align*}
We now note that 
\begin{equation}\label{eq:trace_distance_upper}
1 - \braket{U|\rho_{W_f}|U} \leq D(\rho_{W_f}, \ket{U}\bra{U}).
\end{equation}
Indeed, the projection onto $\ket{U}\bra{U}$ accepts with probability 1 for the state $\ket{U}$, but with probability $\braket{U|\rho_{W_f}|U}$ for the state $\rho_{W_f}$. The difference between these two probabilities is at most the trace distance, $D(\rho_{W_f}, \ket{U}\bra{U})$.

But we recall that $\rho_{W_f}$ corresponds to the pure state
\begin{equation*}
\ket{W_f} = \frac{1}{\sqrt{N}}((W_f J)\otimes I)\sum_{x \in \{0, 1\}^n} \ket{x}\otimes \ket{x}
\end{equation*}
after discarding the $m-n$ ancilla qubits in the first register, while $\ket{U}$ can be obtained from
\begin{equation*}
\frac{1}{\sqrt{N}}(((U\otimes I)J)\otimes I)\sum_{x \in \{0, 1\}^n} \ket{x}\otimes \ket{x}
\end{equation*}
by discarding the same $m-n$ ancilla qubits. 
Thus, $\rho_{W_f}$ and $\ket{U}\bra{U}$ are obtained by applying $W_f \otimes I$ and $(U \otimes I) \otimes I$, respectively, to the same input state
\begin{equation*}
\frac{1}{\sqrt{N}}(J\otimes I)\sum_{x \in \{0, 1\}^n} \ket{x}\otimes \ket{x}.
\end{equation*}
We therefore bound the trace distance between $\rho_{W_f}$ and $\ket{U}\bra{U}$ by the diamond distance of the channels corresponding to $W_f$ and $U$, yielding
\begin{equation*}
D(\rho_{W_f}, \ket{U}\bra{U}) \leq D_\diamond(\pi \circ \Phi_{W_f} \circ i, \Phi_U),
\end{equation*}
which, together with Equation \ref{eq:trace_distance_upper}, proves the lemma.

\end{proof}

\subsection{Proof of Theorem \ref{thm:main_1}}
We will first show using Corollary \ref{cor:net} and Lemma \ref{lem:dirty_to_clean} that a unitary synthesis algorithm can be used to construct a net $\mathcal{T}$ of Hermitian contractions, such that on average, every $U^\dag Q U$ for a fixed Pauli $Q \in \mathcal{P}_n$ and a Haar random unitary $U$ has good trace correlation with \textit{some} element of $\mathcal{T}$. 
Then, to complete the theorem, we will give a lower bound on the size of any net that satisfies this property. 
\begin{lemma}\label{lem:rotated_pauli_correlation_net}
If $V_0, \dots, V_k \in U(M)$ are a random-oracle unitary synthesis algorithm satisfying the conditions of Theorem \ref{thm:main_1}, there is a net $\mathcal{T}$ of Hermitian contractions satisfying $\log_2(|\mathcal{T}|) = O\left(\frac{N k^2}{(1-\epsilon)^4}\right)$ and a nonidentity Pauli operator
$Q \in \mathcal{P}_n\setminus\{I\}$ such that
\begin{equation*}
\Pr_U\left[\max_{C \in \mathcal{T}}
\frac{1}{N}\text{tr}\left(C(U^\dag Q U)\right) \geq \frac{1-\epsilon}{8}\right] \geq \frac{1-\epsilon}{8}.
\end{equation*}
\end{lemma}
\begin{proof}
Say that $V_0, \dots, V_k \in U(M)$ is a random-oracle unitary synthesis algorithm. 
In particular, for
\begin{equation*}
    W_f = V_k\mathcal{O}_f \dots \mathcal{O}_f V_0,
\end{equation*}
we know there is a coupling of the distributions of $\pi \circ \Phi_{W_f} \circ i$ and $\Phi_U$ for a Haar random unitary $U$ such that
\begin{equation*}
\bbE_{U \sim \mathrm{Haar}, W_f}\left[D_\diamond(\pi \circ \Phi_{W_f} \circ i, \Phi_U)\right] \leq \epsilon. 
\end{equation*}
In the remainder of this proof, we consider the joint distribution $(W_f, U)$ that satisfies these properties. 
Writing $A_{f,P}= J^\dag W_{f}^\dag (P \otimes I) W_f J$ for a uniformly random Pauli $P \in \mathcal{P}_n$, we can use Lemma \ref{lem:dirty_to_clean} to write
\begin{equation*}
\bbE_{P \in \mathcal{P}_n} \left[
\frac{1}{N}\text{tr}\left(A_{f,P}(U^\dag P U)\right)
\right]
\geq 1-D_\diamond(\pi \circ \Phi_{W_f} \circ i, \Phi_U),
\end{equation*}
and taking expectations in $W_f$ and $U$,
\begin{equation*}
\bbE_{P \in \mathcal{P}_n, W_f, U} \left[
\frac{1}{N}\text{tr}\left(A_{f,P}(U^\dag P U)\right)
\right] \geq 1-\epsilon.
\end{equation*}
In particular, we may fix some Pauli $Q \in \mathcal{P}_n$ such that the inequality holds. 
However, we must enforce that $Q \neq I$. 
If $Q = I$, then the expression 
\begin{align*}
\frac{1}{N}\text{tr}\left(A_{f,Q}(U^\dag Q U)\right)
&= \frac{1}{N}\text{tr}(I)\\
&= 1.
\end{align*}
Hence, if we are forced to choose $Q$ among $\mathcal{P}_n \backslash \{I\}$, the lower bound becomes
\begin{align*}
\bbE_{W_f, U} \left[
\frac{1}{N}\text{tr}\left(A_{f,Q}(U^\dag Q U)\right)
\right]
&\geq \frac{(1-\epsilon)-N^{-2}}{1-N^{-2}}\\
&\geq \frac{1-\epsilon}{2},
\end{align*}
where the second inequality holds by the theorem assumption that $1-\eps = \omega(N^{-2})$.
Now, we use Corollary \ref{cor:net} on the circuit $W_f^\dag (Q \otimes I) W_f$ to produce a net $\mathcal{T}$ of size
$\log_2|\mathcal{T}| = O\left(\frac{N(2k)^2}{\zeta^2d^2}\right)
= O\left(\frac{N k^2}{\zeta^2d^2}\right)$
such that 
\begin{equation*}
\Pr_{f}\left[
\min_{C \in \mathcal{T}}
\lVert C-A_{f,Q}\rVert_F > d\sqrt{N}
\right] \leq \zeta.
\end{equation*}
We choose the parameters $\zeta = \frac{1-\epsilon}{16}$ and $d = \frac{1-\epsilon}{8}$, and now consider
\begin{equation}\label{eq:expectation_term}
\bbE_{U}\left[
\max_{C \in \mathcal{T}}
\frac{1}{N}\text{tr}\left(C(U^\dag Q U)\right)
\right].
\end{equation}
If we let $E$ be the event that $\min_{C \in \mathcal{T}}\lVert C-A_{f,Q}\rVert_F \leq d\sqrt{N}$, we can lower bound the expression in \eqref{eq:expectation_term} by
\begin{equation*}
\Pr[E]\left(
\bbE_{W_f,U}\left[
\frac{1}{N}\text{tr}\left(A_{f,Q}(U^\dag Q U)\right)
\Bigm| \mathds{1}_E=1
\right]-d
\right)
-(1-\Pr[E]).
\end{equation*}
Here, we are expanding out the conditional expectation, and when the event $E$ does not occur we are using the bound
\begin{equation*}
\left|
\frac{1}{N}\text{tr}\left(C(U^\dag Q U)\right)
\right|
\leq \frac{1}{N}\lVert Q\rVert_1\lVert C\rVert \leq 1,
\end{equation*}
in which we use Hölder's inequality and note that $\lVert U^\dag Q U \rVert_1 = \lVert Q \rVert_1 = N$. 
When $E$ does occur, on the other hand, there is some $C \in \mathcal{T}$ for which
\begin{align*}
\text{tr}\left(C(U^\dag Q U)\right)
&\geq \text{tr}\left(A_{f,Q}(U^\dag Q U)\right)
-\lVert C-A_{f,Q}\rVert_F\lVert U^\dag Q U\rVert_F\\
&\geq \text{tr}\left(A_{f,Q}(U^\dag Q U)\right)-dN,
\end{align*}
where in the first inequality we use Cauchy--Schwarz to bound $\text{tr}\left((C-A_{f,Q})(U^\dag Q U)\right)$.
Since this occurs with probability at least $1-\zeta$, we deduce that the expression in \eqref{eq:expectation_term} has the lower bound
\begin{align*}
&\Pr[E]\left(
\bbE_{W_f,U}\left[
\frac{1}{N}\text{tr}\left(A_{f,Q}(U^\dag Q U)\right)
\Bigm| \mathds{1}_E=1
\right]-d
\right)
-(1-\Pr[E])
\\
&\qquad \geq
\bbE_{W_f,U}\left[
\frac{1}{N}\text{tr}\left(A_{f,Q}(U^\dag Q U)\right)
\right]
-d-2(1-\Pr[E])
\\
&\qquad \geq \frac{1-\epsilon}{2}-d-2\zeta
\\
&\qquad = \frac{1-\epsilon}{4}.
\end{align*}
Since the term within the expectation is at most 1, the above inequality implies that there must be a probability of at least $\frac{1-\epsilon}{8}$ that
\begin{equation*}
\max_{C \in \mathcal{T}}
\frac{1}{N}\text{tr}\left(C(U^\dag Q U)\right) \geq \frac{1-\epsilon}{8},
\end{equation*}
as desired. 
\end{proof}
We now show that, for each fixed Hermitian contraction
$C$ and nonidentity Pauli $Q$, the normalized correlation
$N^{-1}\operatorname{tr}(CU^\dag Q U)$ concentrates
around zero.  
\begin{lemma}\label{lem:haar_concentration}
Let $C$ be a fixed Hermitian contraction, and
$Q \in \mathcal{P}_n$ a nonidentity Pauli operator.
Then for a Haar random $U$ and every $t\ge 0$, 
\begin{equation*}
\Pr\left[\frac{1}{N}\text{tr}\left(C(U^\dag Q U )\right) \geq t\right] \leq \text{exp}\left(\frac{-N^2t^2}{48}\right)
\end{equation*}
\end{lemma}
\begin{proof}
Define $h(U)=N^{-1}\operatorname{tr}(CU^\dag Q U)$. 
We see as before that using Hölder's inequality, we obtain
\begin{align*}
\lvert h(U)-h(V)\rvert
&\leq \frac{1}{N}
\left|\text{tr}(C(U-V)^\dag Q U)\right|
+\frac{1}{N}
\left|\text{tr}(C V^\dag Q(U-V))\right|\\
&\leq \frac{2}{N}\lVert U - V \rVert_F\lVert C \rVert_F\\
&\leq \frac{2}{\sqrt{N}} \lVert U - V \rVert_F.
\end{align*}
In the third inequality, we use that a Hermitian contraction has Frobenius norm at most $\sqrt{N}$. 
Hence, with respect to these two norms, the function $h$ has Lipschitz constant $L=\frac{2}{\sqrt{N}}$.
Note also that $\bbE_U[U^\dag Q U] = 0$ since $\text{tr}(Q) = 0$, implying that $\bbE_U[h(U)] = 0$. Applying Corollary 17 of \cite{MM13}, we obtain 
\begin{align*}
\Pr_U\left[h(U) \geq t\right] &= \Pr_U\left[h(U) \geq \bbE_U [h(U)] + t\right] \\&\leq \text{exp}\left(\frac{-Nt^2}{12L^2}\right)\\
&\leq \text{exp}\left(\frac{-N^2t^2}{48}\right),
\end{align*}
completing the proof. 
\end{proof}

\begin{proof}[Proof of Theorem \ref{thm:main_1}]
Lemma \ref{lem:rotated_pauli_correlation_net} shows that a random-oracle unitary synthesis algorithm $V_0, \dots, V_k$ and error at most $\epsilon$ has a net $\mathcal{T}$ of Hermitian contractions such that $\log_2|\mathcal T|\le\frac{C_0 Nk^2}{(1-\epsilon)^4}$ for a constant $C_0>0$, and 
\begin{equation*}
\Pr_U\left[\max_{C \in \mathcal{T}}
\frac{1}{N}\text{tr}\left(C(U^\dag Q U)\right) \geq \frac{1-\epsilon}{8}\right] \geq \frac{1-\epsilon}{8}. 
\end{equation*}
On the other hand, Lemma \ref{lem:haar_concentration}
shows that, for each fixed $C\in\mathcal T$,
\[
\mathbb P_U\left[
\frac{1}{N}\operatorname{tr}(CU^\dag Q U)
\ge\frac{1-\epsilon}{8}
\right]
\le \text{exp}\left(-cN^2(1-\epsilon)^2\right)
\]
for a constant $c>0$. 
We conclude that by a union bound
\begin{equation*}
\text{exp}\left(
\frac{C_0 Nk^2}{(1-\epsilon)^4}
-cN^2(1-\epsilon)^2
\right)
\ge \frac{1-\epsilon}{8}.
\end{equation*}
This equation is only possible when
\begin{equation*}
k^2(1-\eps)^{-6} =\Omega(N),
\end{equation*}
as desired. 
\end{proof}

\subsection{Proof of Theorem \ref{thm:main_2}}
We will now show our second main theorem, the lower bound against ROM-PRUs.
The crucial steps are very similar, as we still create a small net with good trace correlation---but rather than a single Pauli $Q$, now the net includes pairs of Pauli operators $Q$ and contractions $C$. 
\begin{lemma}\label{lem:rotated_pauli_correlation_list}
If $V_0, \dots, V_k \in U(M)$ are a ROM-PRU satisfying
the conditions of Theorem \ref{thm:main_2}, then for every
$\beta \in (0,1)$, there is a net $\mathcal{T}$ consisting
of pairs $(Q,C)$, where $Q \in \mathcal{P}_n \setminus \{I\}$
and $C$ is a Hermitian contraction, satisfying
\[
\log_2|\mathcal{T}|
=O\left(\frac{N k^2}{\beta^2(1-\epsilon)^4}+\log_2 N\right)
\]
such that
\begin{equation*}
\Pr_{U \sim \nu}\left[
\max_{(Q,C) \in \mathcal{T}}
\frac{1}{N}\text{tr}\left(C(U^\dag Q U)\right)
\geq \frac{1-\epsilon}{8}
\right] \geq 1-\beta,
\end{equation*}
where $\nu$ is a standard non-adaptive unitary $T$-design with error $\delta$,
with $T=\Omega(N^{1+\gamma})$.
\end{lemma}
\begin{proof}
Say that $V_0, \dots, V_k \in U(M)$ is a ROM-PRU
satisfying the conditions of Theorem \ref{thm:main_2}.
In particular, for
\begin{equation*}
    W_f = V_k\mathcal{O}_f \dots \mathcal{O}_f V_0,
\end{equation*}
we know there is a coupling of the distributions of
$\pi \circ \Phi_{W_f} \circ i$ and $\Phi_U$, for
$U \sim \nu$, such that
\begin{equation*}
D_\diamond(\pi \circ \Phi_{W_f} \circ i, \Phi_U)
\leq \epsilon.
\end{equation*}
In the remainder of this proof, we consider the joint distribution $(W_f, U)$ that satisfies these properties. 
Writing $A_{f,P}= J^\dag W_{f}^\dag (P \otimes I) W_f J$ for a uniformly random Pauli $P \in \mathcal{P}_n$, we can use Lemma \ref{lem:dirty_to_clean} to write
\begin{align*}
\bbE_{P \in \mathcal{P}_n} \left[
\frac{1}{N}\text{tr}\left(A_{f,P}(U^\dag P U)\right)
\right]
&\geq 1-D_\diamond(\pi \circ \Phi_{W_f} \circ i, \Phi_U)\\
&\geq 1-\epsilon,
\end{align*}
Since the guarantee holds almost surely under the coupling,
we do not need to take expectations over $U$ and $W_f$.
Excluding $P=I$ as in
Lemma \ref{lem:rotated_pauli_correlation_net}, we obtain
\begin{align*}
\bbE_{Q \in \mathcal{P}_n \setminus \{I\}}\left[
\frac{1}{N}\text{tr}\left(A_{f,Q}(U^\dag Q U)\right)
\right]
&\geq \frac{(1-\epsilon)-N^{-2}}{1-N^{-2}}\\
&\geq \frac{1-\epsilon}{2},
\end{align*}
using the same assumption on $1-\epsilon$.
Each term in the expectation has an upper bound of 1 by
\begin{equation*}
\left|
\frac{1}{N}\text{tr}\left(C(U^\dag Q U)\right)
\right|
\leq \frac{1}{N}\lVert Q\rVert_1\lVert C\rVert \leq 1.
\end{equation*}
It therefore follows that there is a subset of Paulis
$S_{f,U} \subseteq \mathcal{P}_n \setminus \{I\}$ satisfying
$\frac{|S_{f,U}|}{N^2-1} \geq \frac{1-\epsilon}{4}$
such that for any $Q \in S_{f,U}$,
\begin{equation*}
\frac{1}{N}\text{tr}\left(A_{f,Q}(U^\dag Q U)\right)
\geq \frac{1-\epsilon}{4},
\end{equation*}
or else the expectation would be strictly less than
$\frac{1-\epsilon}{2}$.
For every Pauli $Q \in \mathcal{P}_n \setminus \{I\}$,
we now use Corollary \ref{cor:net} on the circuit
$W_f^\dag (Q \otimes I) W_f$ to produce a net
$\mathcal{T}_Q$ of size
$\log_2|\mathcal{T}_Q|
= O\left(\frac{N(2k)^2}{\zeta^2d^2}\right)
= O\left(\frac{N k^2}{\zeta^2d^2}\right)$
such that
\begin{equation*}
\Pr_{f}\left[
\min_{C \in \mathcal{T}_Q}
\lVert C-A_{f,Q}\rVert_F > d\sqrt{N}
\right] \leq \zeta.
\end{equation*}
We choose the parameters
$\zeta = \frac{\beta(1-\epsilon)}{4}$
and $d = \frac{1-\epsilon}{8}$. 

Let $R(f)$ be the fraction of nonidentity Paulis for which
$A_{f,Q}$ is not covered by $\mathcal{T}_Q$,
\begin{equation*}
R(f) =
\Pr_{Q \in \mathcal{P}_n \setminus \{I\}}\left[
\min_{C \in \mathcal{T}_Q}
\lVert C-A_{f,Q}\rVert_F > d\sqrt{N}
\right].
\end{equation*}
We would like to establish that with reasonably high probability over $f$, $R(f)$ is small, i.e. most Pauli operators $Q$ produce operators $A_{f, Q}$ that are near their net $\mathcal{T}_Q$. 
However, 
\begin{align*}
\bbE_f[R(f)] &= \bbE_f\left[\Pr_{Q \in \mathcal{P}_n \setminus \{I\}}\left[
\min_{C \in \mathcal{T}_Q}
\lVert C-A_{f,Q}\rVert_F > d\sqrt{N}
\right]\right]\\
&= \bbE_{Q \in \mathcal{P}_n \setminus \{I\}}\left[\Pr_f\left[
\min_{C \in \mathcal{T}_Q}
\lVert C-A_{f,Q}\rVert_F > d\sqrt{N}
\right]\right]\\
&\leq \zeta. 
\end{align*}
By Markov's inequality,
\begin{equation*}
\Pr_f\left[R(f) \geq \frac{1-\epsilon}{4}\right]
\leq \frac{4\zeta}{1-\epsilon}
= \beta.
\end{equation*}
Thus, with probability at least $1-\beta$ over $f$,
strictly less than an $(1-\epsilon)/4$ fraction of the
nonidentity Paulis are not sufficiently near their net. 
Since the set $S_{f,U}$ occupies a fraction
$\frac{|S_{f,U}|}{N^2-1} \geq \frac{1-\epsilon}{4}$,
there is a $Q \in S_{f,U}$ and
$C \in \mathcal{T}_Q$ such that
$\lVert C-A_{f,Q}\rVert_F \leq d\sqrt{N}$.
For this pair, as before Cauchy-Schwarz gives
\begin{align*}
\frac{1}{N}\text{tr}\left(C(U^\dag Q U)\right)
&\geq
\frac{1}{N}\text{tr}\left(A_{f,Q}(U^\dag Q U)\right)
-\frac{1}{N}
\lVert C-A_{f,Q}\rVert_F
\lVert U^\dag Q U\rVert_F\\
&\geq \frac{1-\epsilon}{4}-d\\
&= \frac{1-\epsilon}{8}.
\end{align*}
We now define the total list
\begin{equation*}
\mathcal{T}
= \bigcup_{Q \in \mathcal{P}_n \setminus \{I\}}
\left(\{Q\}\times\mathcal{T}_Q\right).
\end{equation*}
We conclude that
\begin{equation*}
\Pr_{U \sim \nu}\left[
\max_{(Q,C)\in\mathcal{T}}
\frac{1}{N}\text{tr}\left(C(U^\dag Q U)\right)
\geq \frac{1-\epsilon}{8}
\right] \geq 1-\beta.
\end{equation*}
There are $N^2-1$ nonidentity Paulis, so
\[
\log_2|\mathcal T|
\le\log_2(N^2-1)+\max_{Q\ne I}\log_2|\mathcal T_Q|
=O\left(\frac{N k^2}{\beta^2(1-\epsilon)^4}+\log_2 N\right).
\]
Here we substituted $\zeta=\beta(1-\epsilon)/4$
and $d=(1-\epsilon)/8$.
The union contributes at most $2\log_2 N$ bits to the
logarithm of the list size. 
\end{proof}

\begin{lemma}[
    {\cite[Corollary A.4]{BO24}}]
\label{lem:quantum_threshold_decision}
Let $M_1,\ldots,M_K$ be operators on a finite-dimensional
Hilbert space $\mathcal H$ satisfying $0\le M_j\le I$.
Let $0<\xi,\eta<1/2$ and $0\le b\le 1-\xi$.
There exists a two-outcome measurement $\{B,I-B\}$ on
$\mathcal H^{\otimes t}$, where
\[
t=O\!\left(\frac{\log_2(K/\eta)}{\xi^2}\right),
\]
such that, for every density operator $\rho$ on $\mathcal H$,
\begin{align*}
\max_{j\in[K]}\operatorname{tr}(M_j\rho)\ge b+\xi
&\quad\Longrightarrow\quad
\operatorname{tr}(B\rho^{\otimes t})\ge 1-\eta,\\
\max_{j\in[K]}\operatorname{tr}(M_j\rho)\le b
&\quad\Longrightarrow\quad
\operatorname{tr}(B\rho^{\otimes t})\le \eta.
\end{align*}
\end{lemma}

\begin{proof}[Proof of Theorem \ref{thm:main_2}]
Choose $\beta=\frac{1-\delta}{8}$. By
Lemma \ref{lem:rotated_pauli_correlation_list}, there is a fixed list $\mathcal T$ of pairs $(Q,C)$, with $Q$
a nonidentity Pauli and $C$ a Hermitian contraction,
such that
\[
\log_2|\mathcal T|
=O\left(
\frac{N k^2}{(1-\epsilon)^4(1-\delta)^2}
+\log_2 N
\right),
\]
and
\begin{equation}\label{eq:design_correlation_event}
\mathbb P_{U\sim\nu}\left[
\max_{(Q,C)\in\mathcal T}
\frac{1}{N}\text{tr}(CU^\dag Q U)
\ge \frac{1-\epsilon}{8}
\right]
\ge 1-\beta. 
\end{equation}
We define a measurement for each of the trace correlations in the above expression. 
We again consider the Choi state,
\begin{align*} \ket{U}&=\frac{1}{\sqrt{N}}(U\otimes I)
\sum_{x\in\{0,1\}^n}\ket{x} \otimes \ket{x} \\
&= \frac{1}{\sqrt{N}}(I\otimes U^T)
\sum_{x\in\{0,1\}^n}\ket{x} \otimes \ket{x} .
\end{align*}
We define $\rho_U = \ket{U}\bra{U}$. Applying $U^{\otimes t}$ in parallel to the first halves of $t$ maximally entangled pairs prepares $\rho_U^{\otimes t}$.
For each $(Q,C)\in\mathcal T$, define
\begin{equation*}
M_{Q,C}=\frac{I+Q\otimes C^T}{2}.
\end{equation*}
Since $Q\otimes C^T$ is a Hermitian contraction, our operator is Hermitian and $0\le M_{Q,C}\le I$. Moreover, we can expand
\begin{align*}
\text{tr}((Q \otimes C^T)\rho_U)&= \text{tr}((Q \otimes I)\rho_U(I \otimes C^T))\\
&= \frac{1}{N}\operatorname{tr}\left[
(QU \otimes I)
\left(\sum_{x,y \in \{0,1\}^n}
\ket{x}\bra{y} \otimes \ket{x}\bra{y}\right)
\left(I \otimes \overline{UC}\right)
\right]\\
&= \frac{1}{N}\sum_{x, y \in \{0, 1\}^n} (QU)_{yx}(UC)_{yx}^*\\
&=\frac{1}{N}\text{tr}(CU^\dag Q U),
\end{align*}
where the first equality uses cyclicity of trace and the last uses that $C$ is Hermitian so $(UC)^\dag = CU^\dag$. 
It follows that 
\begin{equation*}
\text{tr}(M_{Q, C} \rho_U) = \frac{1}{2} + \frac{1}{2N}\text{tr}(CU^\dag Q U). \end{equation*}
Hence, these observables correspond directly to the trace correlations we are analyzing.
We now apply Lemma \ref{lem:quantum_threshold_decision} to assemble one measurement $B$ that will operationalize the expression $\max_{(Q,C)\in\mathcal T}
\frac{1}{N}\text{tr}(CU^\dag Q U)$. Choose $\xi=\frac{1-\epsilon}{32}$, $b=\frac{1}{2}+\xi$, $\eta=\frac{1-\delta}{8}$, and $K= |\mathcal{T}|$.

The lemma gives a measurement $\{B,I-B\}$ on $t=O\left(\frac{\log_2(K/\eta)}{(1-\epsilon)^2}\right)$ Choi states, satisfying
\begin{align*}
\max_{(Q,C)\in\mathcal T}
\frac{1}{N}\text{tr}(CU^\dag Q U)\ge\frac{1-\epsilon}{8}
&\implies
\operatorname{tr}(B\rho_U^{\otimes t})\ge 1-\eta,\\
\max_{(Q,C)\in\mathcal T}
\frac{1}{N}\operatorname{tr}(CU^\dag Q U)\le \frac{1-\epsilon}{16}
&\implies
\operatorname{tr}(B\rho_U^{\otimes t})\le \eta.
\end{align*}
These implications follow by applying the lemma with
the specified parameters and using the identity
\[
\operatorname{tr}(M_{Q,C}\rho_U)
=\frac12+\frac{1}{2N}
\operatorname{tr}(CU^\dag Q U).
\]
As per the lemma, this implication is independent of $U$, since the implication on observables is independent of the Choi state the measurement is applied on. 

Using Equation \ref{eq:design_correlation_event}, we obtain that
\begin{align*}
\mathbb E_{U\sim\nu}
\operatorname{tr}(B\rho_U^{\otimes t})
&\ge (1-\beta)(1-\eta)\\
&=\left(1-\frac{1-\delta}{8}\right)^2\\
&\geq 1 - \frac{1-\delta}{4}. 
\end{align*}
since with probability $1 - \beta$, the lower bound of $1-\eta$ follows by the guarantee of the lemma. 
On the other hand, Lemma \ref{lem:haar_concentration}
gives, for each fixed $(Q,C)\in\mathcal T$,
\begin{equation*}
\mathbb P_{U\sim\mathrm{Haar}}\left[
\frac{1}{N}\operatorname{tr}(CU^\dag Q U)>\frac{1-\epsilon}{16}
\right]
\leq \text{exp}\left(-cN^2\left(1-\epsilon\right)^2\right)
\end{equation*}
for some $c>0$.
A union bound yields
\begin{equation*}
\Pr\left[\max_{(Q,C)\in\mathcal T}
\frac{1}{N}\operatorname{tr}(CU^\dag Q U)\le \frac{1-\epsilon}{16}\right] \geq 1 - K\text{exp}\left(-cN^2(1-\epsilon)^2\right),
\end{equation*}
and as a consequence
\begin{align*}
\bbE_{U \sim \mathrm{Haar}}\left[\text{tr}(B\rho_U^{\otimes t})\right] &\leq \eta + K\text{exp}\left(-cN^2(1-\epsilon)^2\right)\\
&\leq \frac{1-\delta}{8} + K\text{exp}\left(-cN^2(1-\epsilon)^2\right),
\end{align*}
since $\text{tr}(B\rho_U^{\otimes t}) \leq 1$ is greater than $\eta$ with probability at most $K\text{exp}\left(-cN^2(1-\epsilon)^2\right)$. 

We may assume that $t\leq T$, where
$t=O\left(\frac{\log_2(K/\eta)}{\xi^2}\right)$
is given by Lemma \ref{lem:quantum_threshold_decision}. 
Indeed, otherwise our upper bound on $K = |\mathcal{T}|$ yields
\begin{align*}
t
&=O\left(
\frac{N k^2}{(1-\delta)^2(1-\epsilon)^6}
+\frac{\log_2 N+\log_2(8/(1-\delta))}{(1-\epsilon)^2}
\right)\\
&= O\left(
\frac{N k^2}{(1-\delta)^2(1-\epsilon)^6}\right).
\end{align*}
Since otherwise $t>T$ and
$T=\Omega(N^{1+\gamma})$, rearranging gives
\[
k^2
=\Omega\left(
(1-\delta)^2(1-\epsilon)^6\frac{T}{N}
\right)
=\Omega\left(
(1-\delta)^2(1-\epsilon)^6 N^{\gamma}
\right).
\]
which already implies the desired bound
\[
k^2=\Omega\left(
(1-\delta)^2(1-\epsilon)^6N^{\min(1, \gamma)}
\right),
\]

The $t$ forward queries are parallel and followed by the joint measurement
$\{B,I-B\}$. Since $t\leq T$, the non-adaptive design guarantee implies
\begin{equation*}
\left|
\mathbb E_{U\sim\nu}
\operatorname{tr}(B\rho_U^{\otimes t})
-
\mathbb E_{U\sim\mathrm{Haar}}
\operatorname{tr}(B\rho_U^{\otimes t})
\right|
\le\delta.
\end{equation*}
Hence, we can chain the upper and lower bounds on $\mathbb E_{U\sim\nu}
\operatorname{tr}(B\rho_U^{\otimes t})$ to obtain 
\begin{equation*}
1 - \frac{1-\delta}{4}
\le
\delta+\frac{1-\delta}{8}+K\text{exp}\left(-cN^2(1-\epsilon)^2\right).
\end{equation*}
Rearranging gives
\[
K\text{exp}\left(-cN^2(1-\epsilon)^2\right)
\ge \frac{5(1-\delta)}{8},
\]
so
\[
\log_2 K
\ge cN^2(1-\epsilon)^2
-\log_2\frac{8}{5(1-\delta)}.
\]
Combining this with our upper bound on $\log_2 K$ gives
\[
N^2(1-\epsilon)^2
=O\left(
\frac{N k^2}{(1-\delta)^2(1-\epsilon)^4}
+\log_2 N+\log_2\frac{8}{5(1-\delta)}
\right).
\]
Since $k\ge1$ and $0<1-\epsilon,1-\delta\le1$,
the logarithmic terms are absorbed into the first term.
Rearranging yields
\[
k^2(1-\delta)^{-2}(1-\epsilon)^{-6}
=\Omega(N).
\]

\end{proof}

\section{An \texorpdfstring{$O(N)$}{O(N)}-query secure ROM-PRU construction} \label{sec:linear_design}
In this section, we will demonstrate a ROM-PRU construction that is secure against $O(N)$ queries. 
The previously best known result was security against $O(\sqrt{N})$ queries. 
Combined with Theorem \ref{thm:main_2}, this provides a tight characterization of the best possible designs achievable in the ROM-PRU model.
Unfortunately, the construction we provide here involves fixed unitaries $H_1, \dots, H_n$ that are not efficiently implementable. 
However, we believe this result provides evidence that the current $O(\sqrt{N})$-query design barrier is primarily a technical barrier, even when the fixed unitaries $H_i$ are required to be efficient. 

We now define our unitary design construction. For any tuple of $n$ unitaries $H := (H_1, \ldots, H_n)$ which each act on $n$ qubits, we define the unitary design $\mu_H$ as follows:
\begin{enumerate}
    \item Sample $n+1$ random functions $f_i: [N] \to [N]$.
    \item Output $U = F_n H_n \ldots F_1 H_1 F_0$, where $F_j = \sum_{x \in [N]} e^{2 \pi i f_j(x)/N} \ketbra{x}{x}$.
\end{enumerate}
\begin{theorem}\label{thm:linear_design}
With probability at least $1-2^{-\Omega(n^2)}$ over the choice of Haar-random unitaries $H_1, \ldots, H_n$, the ensemble $\mu_{(H_1, \ldots, H_n)}$ is a $2^{-n^2/3}$-approximate strong adaptive unitary $(N/8)$-design in the sense of Definition~\ref{def:strong_adaptive_design}.
\end{theorem}
\begin{proof}
Let $T := N/8$ be a bound on the maximum number of queries. As our PRU construction is invariant under a random global phase $\phi \gets \{1, \omega_N, \ldots, \omega_N^{N-1}\}$ and $T < N$, by \cite[Theorem 1.1]{TW25} we can decontrol all calls to $\{cU, cU^{\dagger}, c\overline{U}, cU^T\}$.

Fix (for now) the choice of $H = (H_1, \ldots, H_n)$. For any query to the unitary $U$, note that the channel 
    \[ (UXU^{\dagger})_{ij} = \sum_{a, b} U_{ia} \overline{U}_{jb} X_{ab} \]
is composed of terms which have exactly one entry corresponding to $U$ and one entry corresponding to $\overline{U}$. The same holds for queries to $U^{\dagger}$, $\overline{U}$, and $U^T$. This means that the acceptance probability $f_{\CA}(U)$ of any $T$-query algorithm $\CA$ is a linear combination of terms of the form
    \[ \prod_{r = 1}^q U_{i_ra_r} \overline{U}_{j_rb_r}, \]
where $0 \leq q \leq T$, which are precisely the entries of the matrix $\Gamma_q(U) := U^{\otimes q} \otimes \overline{U}^{\otimes q}$. Thus, we have that 
    \[ f_{\mathcal{A}} \in \spn\{U \to \Gamma_q(U)_{ij}: i, j, 0 \leq q \leq T\}. \]
By mixed Schur-Weyl duality, we know that the entries of $\Gamma_q$ are linear combinations of entries of $\rho_{\lambda, \nu}$ where $\lambda, \nu \vdash r$ for $0 \leq r \leq q$, so $f_{\mathcal{A}}$ must lie in the span of the entries of $\rho_{\lambda, \nu}$ where $\lambda, \nu \vdash q$ and $0 \leq q \leq T$.

Observing that the Schur orthogonality relations imply that
    \[ \int_{U(N)} \rho_{\lambda, \nu}(U)_{ij} \overline{\rho_{\lambda',\nu'}(U)_{i'j'}} dU = \frac{\delta_{\lambda, \lambda'}\delta_{\nu, \nu'}\delta_{i, i'}\delta_{j, j'}}{\dim \rho_{\lambda, \nu}}, \]
we see that the functions
    \[ \phi_{\lambda, \nu, i, j}(U) = \sqrt{d_{\lambda, \nu}} \rho_{\lambda, \nu}(U)_{ij}, \]
where $d_{\lambda, \nu} := \dim \rho_{\lambda, \nu}$, are orthonormal. We can therefore decompose $f_{\CA}(U)$ as
    \[ f_{\CA}(U) = \sum_{q = 0}^T \sum_{\lambda, \nu \vdash q} \sum_{i, j = 1}^{d_{\lambda, \nu}} c_{\lambda, \nu, i, j} \sqrt{d_{\lambda, \nu}} \, \rho_{\lambda, \nu}(U)_{ij} \]
for some coefficients $c_{\lambda, \nu, i, j}$. By orthogonality, we know that since $0 \leq f_{\CA}(U) \leq 1$ for all $U$,
    \[ \sum_{q = 0}^T \sum_{\lambda, \nu \vdash q} \sum_{i, j = 1}^{d_{\lambda, \nu}} |c_{\lambda, \nu, i, j}|^2 = \int |f_{\CA}(U)|^2 d\mu_{\mathsf{Haar}}(U) \leq 1. \]
Additionally, since $\rho_{\emptyset, \emptyset} = 1$, orthogonality implies that
    \[ \int_{U(N)} \rho_{\lambda, \nu}(U)_{ij} dU = 0 \]
for all $(\lambda, \nu) \neq (\emptyset, \emptyset)$. Defining 
    \[ \widehat{\mu_H}(\lambda, \nu) := \int \rho_{\lambda, \nu}(U) d\mu_H(U) = \bbE_{F_0, \ldots, F_n}[\rho_{\lambda, \nu}(F_nH_n \ldots F_1H_1F_0)], \]
we can now bound the distinguishing advantage of $\CA$ by
\begin{align*}
    &\left|\int f_{\CA}(U) d\mu_H(U) - \int f_{\CA}(U) d\mu_{\mathsf{Haar}}(U)\right| \\
    &\quad = \left|\sum_{q = 0}^T \sum_{\lambda, \nu \vdash q} \sum_{i, j = 1}^{d_{\lambda, \nu}} c_{\lambda, \nu, i, j} \sqrt{d_{\lambda, \nu}} \left[\widehat{\mu_H}(\lambda, \nu)_{ij}  - \widehat{\mu_{\mathsf{Haar}}}(\lambda, \nu)_{ij}\right]\right| \\
    &\quad = \left|\sum_{q = 0}^T \sum_{\lambda, \nu \vdash q} \sum_{i, j = 1}^{d_{\lambda, \nu}} c_{\lambda, \nu, i, j} \sqrt{d_{\lambda, \nu}} \widehat{\mu_H}(\lambda, \nu)_{ij}  - c_{\emptyset, \emptyset, 1, 1} \widehat{\mu_{\mathsf{Haar}}}(\emptyset, \emptyset)_{11}\right| \\
    &\quad = \left|\sum_{q = 1}^T \sum_{\lambda, \nu \vdash q} \sum_{i, j = 1}^{d_{\lambda, \nu}} c_{\lambda, \nu, i, j} \sqrt{d_{\lambda, \nu}} \widehat{\mu_H}(\lambda, \nu)_{ij}\right| \\
    &\quad \leq \left(\sum_{q = 1}^T \sum_{\lambda, \nu \vdash q} \sum_{i, j = 1}^{d_{\lambda, \nu}} |c_{\lambda, \nu, i, j}|^2 \right)^{1/2} \left(\sum_{q = 1}^T \sum_{\lambda, \nu \vdash q} \sum_{i, j = 1}^{d_{\lambda, \nu}} d_{\lambda, \nu} |\widehat{\mu_H}(\lambda, \nu)_{ij}|^2 \right)^{1/2} \\
    &\quad \leq \left(\sum_{q = 1}^T \sum_{\lambda, \nu \vdash q} d_{\lambda, \nu} \|\widehat{\mu_H}(\lambda, \nu)\|_F^2 \right)^{1/2}.
\end{align*}
As this inequality holds regardless of the coefficients $c_{\lambda, \nu, i, j}$, the distinguishing advantage of any $T$-query algorithm is upper bounded by
    \[ \eta_T(\mu_H) := \left(\sum_{q = 1}^T \sum_{\lambda, \nu \vdash q} d_{\lambda, \nu} \|\widehat{\mu_H}(\lambda, \nu)\|_F^2\right)^{1/2}. \]

Consequently, our new goal will be to bound $\bbE_H[\eta_T(\mu_H)^2]$. We start by defining the projectors 
    \[ P_{\lambda, \nu} := \int_{U(1)^N} \rho_{\lambda, \nu}(U) dU, \]
onto the subspace of vectors in $V^N_{\lambda, \nu}$ which are fixed by all elements in $U(1)^N$; let $b_{\lambda, \nu} := \rank P_{\lambda, \nu}$. In the following lemma we derive some crude bounds on $b_{\lambda, \nu}$ and $d_{\lambda, \nu}$, which we will thereafter relate to $\bbE_H[\eta_T(\mu_H)^2]$.
\begin{lemma}\label{lemma:rank_dim_est}
    For partitions $\lambda, \nu \vdash q$ where $2q < N$, let $f^{\lambda}$ and $f^{\nu}$ denote the number of standard Young tableaux of shape $\lambda$ and $\nu$, respectively. Then,
        \[ b_{\lambda, \nu} \leq \binom{N+q-1}{q} f^{\lambda}f^{\nu}, \]
    and
        \[ d_{\lambda, \nu} \geq \frac{f^{\lambda}f^{\nu}(N-2q)^{2q}}{(q!)^2}. \]
\end{lemma}
\begin{proof}[Proof of Lemma \ref{lemma:rank_dim_est}]
    As the rank of a projection is equal to the trace of the associated matrix, we have that
    \begin{align*}
        b_{\lambda, \nu} &= \rank P_{\lambda, \nu} = \Tr\left(\int_{U(1)^N} \rho_{\lambda, \nu}(U) dU\right) = \int_{U(1)^N} \Tr(\rho_{\lambda, \nu}(U)) dU \\
        &= \int_{z_1, \ldots, z_N \sim U(1)} s_{\lambda, \nu}(z_1, \ldots, z_N) dz \\
        &= \int_{z_1, \ldots, z_N \sim U(1)} \sum_{(T, T') \in RT(\lambda, \nu, N)} \prod_{i = 1}^N z_i^{w_i(T)-w_i(T')} dz \\
        &= \sum_{(T, T') \in RT(\lambda, \nu, N)} \left(\int_{z_1, \ldots, z_N \sim U(1)} \prod_{i = 1}^N z_i^{w_i(T)-w_i(T')} dz\right) \\
        &= \sum_{(T, T') \in RT(\lambda, \nu, N)} \mathds{1}_{w(T) = w(T')} \leq \sum_{T \in SSYT(\lambda, N)} \sum_{T' \in SSYT(\nu, N)} \mathds{1}_{w(T) = w(T')}.
    \end{align*}
    Fix a tuple $\alpha = (\alpha_1, \ldots, \alpha_N)$ such that $\sum_i \alpha_i = q$. The number of semistandard Young tableaux of shape $\lambda$ with weight $\alpha$ is at most the number of standard Young tableaux of shape $\lambda$,\footnote{This is since standardization provides an injective map from SSYT of shape $\lambda$ with weight $\alpha$ to SYT of shape $\lambda$.} so we have that
        \[ b_{\lambda, \nu} \leq \sum_{\alpha} f^{\lambda} f^{\nu} = \binom{N+q-1}{q} f^{\lambda} f^{\nu}, \]
    as there are exactly $\binom{N+q-1}{q}$ possible weights.

    We now bound $d_{\lambda, \nu} = \dim V^N_{\lambda, \nu}$, which is equal to the number of rational tableaux of shape $(\lambda, \nu)$ with alphabet $[N]$. Observing that $(A, B)$ is automatically a rational tableau if all $2q$ entries across $A$ and $B$ are distinct, we can lower bound $d_{\lambda, \nu}$ by the number of pairs of SSYT of shapes $\lambda$ and $\nu$ in which all $2q$ labels are distinct. This in turn can be computed by first selecting $\binom{N}{q} \cdot \binom{N-q}{q}$ distinct labels for the first and second tableaux; after fixing the label set $i_1 < \cdots < i_q$ for each tableau, replacing $i_j$ by $j$ gives a bijection between tableaux using those labels and standard Young tableaux with alphabet $[q]$. We conclude that
        \[ d_{\lambda, \nu} \geq f^{\lambda} f^{\nu} \binom{N}{q} \binom{N-q}{q} \geq \frac{f^{\lambda} f^{\nu}(N-2q)^{2q}}{(q!)^2}, \]
    assuming $2q < N$.
\end{proof}
We now consider the expression $\widehat{\mu_H}(\lambda, \nu)$, which can be written as
\begin{align*}
    \widehat{\mu_H}(\lambda, \nu) &= \bbE_{F_0, \ldots, F_n}\bigl[\rho_{\lambda, \nu}(F_n H_n \cdots F_1 H_1 F_0)\bigr] \\
    &= \bbE_{F_0, \ldots, F_n}\bigl[\rho_{\lambda, \nu}(F_n)\rho_{\lambda, \nu}(H_n)\cdots \rho_{\lambda, \nu}(F_1)\rho_{\lambda, \nu}(H_1)\rho_{\lambda, \nu}(F_0)\bigr] \\
    &= \bbE_{F_n}\bigl[\rho_{\lambda, \nu}(F_n)\bigr] \rho_{\lambda, \nu}(H_n) \cdots \bbE_{F_1}\bigl[\rho_{\lambda, \nu}(F_1)\bigr] \rho_{\lambda, \nu}(H_1)\bbE_{F_0}\bigl[\rho_{\lambda, \nu}(F_0)\bigr],
\end{align*}
since the $F_i$'s are independent and $H_1, \ldots, H_n$ are fixed. 

We claim that $\bbE_{F_i}[\rho_{\lambda, \nu}(F_i)] = P_{\lambda, \nu}$, which will simplify our expression to 
    \[ \widehat{\mu_H}(\lambda, \nu) = P_{\lambda, \nu} \rho_{\lambda, \nu}(H_n) \cdots P_{\lambda, \nu} \rho_{\lambda, \nu}(H_1) P_{\lambda, \nu}. \]
\begin{claim}\label{claim:discrete_to_continuous}
    $\bbE_{F_i}[\rho_{\lambda, \nu}(F_i)] = P_{\lambda, \nu}$.
\end{claim}
\begin{proof}[Proof of Claim \ref{claim:discrete_to_continuous}]
    Since $P_{\lambda, \nu} := \int_{U(1)^N} \rho_{\lambda, \nu}(U) dU$, our goal is to show that the averages of $\rho_{\lambda, \nu}$ over continuous and discrete phases coincide. We do so by showing that these two averages are in fact the same on the entire Hilbert space $(\bbC^N)^{\otimes q} \otimes (\overline{\bbC^N})^{\otimes q}$.
    
    Fix any basis vector
        \[ e_{a_1} \otimes \cdots \otimes e_{a_q} \otimes \overline{e_{b_1}} \otimes \cdots \otimes \overline{e_{b_q}}. \]
    Given a diagonal unitary $D = \sum_x z_x \ketbra{x}{x}$, we have that 
        \[ (D^{\otimes q} \otimes \overline{D}^{\otimes q})(e_{a_1} \otimes \cdots \otimes e_{a_q} \otimes \overline{e_{b_1}} \otimes \cdots \otimes \overline{e_{b_q}}) = \prod_{j = 1}^q z_{a_j} \overline{z_{b_j}} (e_{a_1} \otimes \cdots \otimes e_{a_q} \otimes \overline{e_{b_1}} \otimes \cdots \otimes \overline{e_{b_q}}). \]
When $F_i$ is drawn according to the distribution described in $\mu_H$ and hence each $z_j$ is chosen uniformly from $\{1, \omega_N, \ldots, \omega_N^{N-1}\}$, it is straightforward to see that since $q \leq T < N$, $\bbE_{f_i}[\Gamma_q(F_i)]$ projects onto the following subspace, where the equality of labels denotes multiset equality:
    \[ \spn\{e_{a_1} \otimes \cdots \otimes e_{a_q} \otimes \overline{e_{b_1}} \otimes \cdots \otimes \overline{e_{b_q}}: \{\!\{a_1, \ldots, a_q\}\!\} = \{\!\{b_1, \ldots, b_q\}\!\}\}. \]
Similarly, it is easy to see that when the phases in the diagonal unitary $F_i$ are drawn continuously from $[0, 2\pi)$, the average $\bbE[\Gamma_q(F_i)]$ is the \emph{same} projector. Restricting this equality to $\rho_{\lambda,\nu}$ then finishes the claim.
\end{proof}
We will now think about $H_1, \ldots, H_n$ as Haar-random unitaries. Before we analyze $\widehat{\mu_H}(\lambda, \nu)$, we will need to use the following identity:
\begin{fact}\label{fact:average_identity}
    $\bbE_{U \sim \mu_{\mathsf{Haar}}(N)}[\rho_{\lambda, \nu}(U)^{\dagger} P_{\lambda, \nu} \rho_{\lambda, \nu}(U)] = \frac{b_{\lambda, \nu}}{d_{\lambda, \nu}} I_{d_{\lambda, \nu}}$.
\end{fact}
\begin{proof}[Proof of Fact \ref{fact:average_identity}]
Let $A = \bbE_{U \sim \mu_{\mathsf{Haar}}(N)}[\rho_{\lambda, \nu}(U)^{\dagger} P_{\lambda, \nu} \rho_{\lambda, \nu}(U)]$. For any fixed unitary $G \in U(N)$, 
\begin{align*}
    \rho_{\lambda, \nu}(G)^{\dagger} A \rho_{\lambda, \nu}(G) &= \bbE_{U \sim \mu_{\mathsf{Haar}}(N)}[\rho_{\lambda, \nu}(G)^{\dagger} \rho_{\lambda, \nu}(U)^{\dagger} P_{\lambda, \nu} \rho_{\lambda, \nu}(U) \rho_{\lambda, \nu}(G)] \\
    &= \bbE_{U \sim \mu_{\mathsf{Haar}}(N)}[\rho_{\lambda, \nu}(UG)^{\dagger} P_{\lambda, \nu} \rho_{\lambda, \nu}(UG)] \\
    &= A.
\end{align*}
Thus, $A$ commutes with $\rho_{\lambda, \nu}(G)$ for all $G$. Since $\rho_{\lambda, \nu}$ is irreducible, Schur's lemma implies that
    \[ A = \frac{\Tr A}{d_{\lambda, \nu}} I_{d_{\lambda, \nu}} =  \frac{\bbE_{U \sim \mu_{\mathsf{Haar}}(N)}\left[\Tr\left(\rho_{\lambda, \nu}(U)^{\dagger} P_{\lambda, \nu} \rho_{\lambda, \nu}(U)\right)\right]}{d_{\lambda, \nu}} I_{d_{\lambda, \nu}} = \frac{\Tr P_{\lambda, \nu}}{d_{\lambda, \nu}} I_{d_{\lambda, \nu}} = \frac{b_{\lambda, \nu}}{d_{\lambda, \nu}} I_{d_{\lambda, \nu}}. \]
\end{proof}
Define the partial products $M_0 = P_{\lambda, \nu}$ and $M_r = P_{\lambda, \nu} \rho_{\lambda, \nu}(H_r) M_{r-1}$; by definition $M_n = \widehat{\mu_H}(\lambda, \nu)$. Recalling the identity $\|M\|_F^2 = \Tr(M M^{\dagger})$ and using Fact \ref{fact:average_identity}, we can compute
\begin{align*}
    \bbE_{H_1, \ldots, H_r \sim \mu_{\mathsf{Haar}}(N)}[\|M_r\|_F^2] &= \bbE_{H_1, \ldots, H_r \sim \mu_{\mathsf{Haar}}(N)}[\Tr(M_r M_r^{\dagger})] \\
    &= \bbE[\Tr(P_{\lambda, \nu} \rho_{\lambda, \nu}(H_r) M_{r-1} M_{r-1}^{\dagger} \rho_{\lambda, \nu}(H_r)^{\dagger} P_{\lambda, \nu}^{\dagger})] \\
    &= \bbE[\Tr(\rho_{\lambda, \nu}(H_r)^{\dagger}P_{\lambda, \nu}^{\dagger}P_{\lambda, \nu} \rho_{\lambda, \nu}(H_r) M_{r-1} M_{r-1}^{\dagger})] \\
    &= \bbE[\Tr(\rho_{\lambda, \nu}(H_r)^{\dagger} P_{\lambda, \nu} \rho_{\lambda, \nu}(H_r) M_{r-1} M_{r-1}^{\dagger})] \\
    &= \bbE\left[\Tr\left(\frac{b_{\lambda, \nu}}{d_{\lambda, \nu}} M_{r-1} M_{r-1}^{\dagger}\right)\right] \\
    &= \frac{b_{\lambda, \nu}}{d_{\lambda, \nu}} \cdot\bbE\left[\|M_{r-1}\|_F^2\right].
\end{align*}
By induction we conclude that 
    \[ \bbE[\|\widehat{\mu_H}(\lambda, \nu)\|_F^2] = \bbE[\|M_n\|_F^2] = \left(\frac{b_{\lambda, \nu}}{d_{\lambda, \nu}}\right)^n \|P_{\lambda, \nu}\|_F^2 = \Tr P_{\lambda, \nu} \cdot \left(\frac{b_{\lambda, \nu}}{d_{\lambda, \nu}}\right)^n = b_{\lambda, \nu} \cdot \left(\frac{b_{\lambda, \nu}}{d_{\lambda, \nu}}\right)^n. \]
Putting everything together, we see that
    \[ \bbE_H[\eta_T(\mu_H)^2] = \sum_{q = 1}^T \sum_{\lambda, \nu \vdash q} d_{\lambda, \nu} \bbE_H[\|\widehat{\mu_H}(\lambda, \nu)\|_F^2] = \sum_{q = 1}^T \sum_{\lambda, \nu \vdash q} b_{\lambda, \nu}^2 \left(\frac{b_{\lambda, \nu}}{d_{\lambda, \nu}}\right)^{n-1}. \]
We conclude by showing that this sum is dominated by the $q = 1$ term. 
\begin{lemma}\label{lemma:fourier_concentration}
    $\bbE_H[\eta_T(\mu_H)^2] \leq 2^{-n^2+O(n)}$.
\end{lemma}
\begin{proof}[Proof of Lemma \ref{lemma:fourier_concentration}]
Write $I_q := \sum_{\lambda, \nu \vdash q} b_{\lambda, \nu}^2 \left(\frac{b_{\lambda, \nu}}{d_{\lambda, \nu}}\right)^{n-1}$. We begin by bounding $I_1$. By Lemma \ref{lemma:rank_dim_est}, $b_{(1), (1)} \leq N (f^{(1)})^2 = N$ and $d_{(1), (1)} \geq (N-2)^2 (f^{(1)})^2 = (N-2)^2$, so 
\begin{align*}
    \sum_{\lambda, \nu \vdash 1} b_{\lambda, \nu}^2 \left(\frac{b_{\lambda, \nu}}{d_{\lambda, \nu}}\right)^{n-1} &\leq N^2 \left(\frac{N}{(N-2)^2}\right)^{n-1} \\
    &= 2^{2n} \left(2^{-n}(1-2^{1-n})^{-2}\right)^{n-1} \\
    &= 2^{-n^2+3n} (1-2^{1-n})^{-2(n-1)}.
\end{align*}
For $n > 1$, $(1-2^{1-n})^{-2(n-1)} \leq 4$, so we can bound $I_1$ by $2^{-n^2+3n+2}$.

Applying Lemma \ref{lemma:rank_dim_est} shows that
\begin{align*}
    \sum_{\lambda, \nu \vdash q} b_{\lambda, \nu}^2 \left(\frac{b_{\lambda, \nu}}{d_{\lambda, \nu}}\right)^{n-1} &\leq \sum_{\lambda, \nu \vdash q} b_{\lambda, \nu}^2 \left(\frac{\binom{N+q-1}{q}(q!)^2}{(N-2q)^{2q}}\right)^{n-1} \\
    &\leq \left(\frac{\binom{N+q-1}{q}(q!)^2}{(N-2q)^{2q}}\right)^{n-1} \sum_{\lambda, \nu \vdash q} \left(\binom{N+q-1}{q} f^{\lambda} f^{\nu} \right)^2 \\
    &\leq \left(\frac{\binom{N+q-1}{q}(q!)^2}{(N-2q)^{2q}}\right)^{n-1} \binom{N+q-1}{q}^2 \left(\sum_{\lambda \vdash q} (f^{\lambda})^2 \right)^2.
\end{align*}
Since the Robinson–Schensted–Knuth correspondence \cite{Sta99} implies that $\sum_{\lambda \vdash q} (f^{\lambda})^2 = q!$, this gives the bound
    \[ I_q \leq \left(\frac{\binom{N+q-1}{q}(q!)^2}{(N-2q)^{2q}}\right)^{n-1} \binom{N+q-1}{q}^2 (q!)^2. \]
Since $q \leq T = N/8$, the AM-GM inequality gives
    \[ \binom{N+q-1}{q} q! = N(N+1)\cdots(N+q-1) \leq \left(N+\frac{q-1}{2}\right)^q \leq \left(\frac{17N}{16}\right)^q, \]
so $\binom{N+q-1}{q}^2 (q!)^2 \leq (2N^2)^q$. Additionally, as $N - 2q \geq 3N/4$,
    \[ \frac{\binom{N+q-1}{q}(q!)^2}{(N-2q)^{2q}} \leq q! \left(\frac{17}{9N}\right)^q, \]
and thus
    \[ I_q \leq B_q := (2N^2)^q \left[q!\left(\frac{17}{9N}\right)^q\right]^{n-1}. \]
For all $2 \leq q \leq T$,
    \[ \frac{B_q}{B_{q-1}} = 2N^2\left(\frac{17q}{9N}\right)^{n-1} \leq 2N^2\left(\frac{17}{72}\right)^{n-1} = \frac{144}{17}\left(\frac{17}{18}\right)^n. \]
If $n > 50$ then $\frac{144}{17}\left(\frac{17}{18}\right)^n < 1/2$ and so for sufficiently large $n$,
    \[ \sum_{q = 2}^T I_q \leq \sum_{q = 2}^T B_q \leq B_2 \sum_{q = 0}^{\infty} \left(\frac{1}{2}\right)^q = 2B_2 \leq 8N^4\left(\frac{8}{N^2}\right)^{n-1} = 2^{-2n^2+9n}. \]
Combining this with the bound on $I_1$ gives us
    \[ \bbE_H[\eta_T(\mu_H)^2] = \sum_{q = 1}^T I_q \leq 2^{-n^2+3n+2} + 2^{-2n^2+9n} = 2^{-n^2+O(n)}, \]
as desired.
\end{proof}
By Markov's inequality, we conclude that
    \[ \Pr_H[\eta_T(\mu_H) > 2^{-n^2/3}] \leq \frac{\bbE_H[\eta_T(\mu_H)^2]}{2^{-2n^2/3}} \leq \frac{2^{-n^2+O(n)}}{2^{-2n^2/3}} = 2^{-\Omega(n^2)}. \]
Theorem \ref{thm:linear_design} follows immediately from the fact that the maximum distinguishing advantage of any $T = N/8$-query algorithm is bounded by $\eta_T(\mu_H)$. 
\end{proof}

\section{Acknowledgments}
John Wright is supported by the
NSF CAREER award CCF-233971 and a Sloan fellowship. This work was done in part while some of the authors were visiting the Simons Institute for the Theory of Computing.

\printbibliography
\appendix

\end{document}